\documentclass[11pt]{article}
\usepackage{color}
\usepackage{epsfig}
\usepackage{subfig}
\usepackage{mathrsfs}
\usepackage{booktabs}
\usepackage{graphicx}
\usepackage{epstopdf}
\usepackage{multirow}
\usepackage{amsmath,amssymb}
\usepackage{amsthm}
\usepackage{fancyhdr}
\usepackage{tabularx}
\usepackage{xfrac}
\usepackage{CJK}
\usepackage{mathptmx}
\usepackage{mathrsfs}
\usepackage{graphics}
\usepackage{times}

\newtheorem{theorem}{Theorem}[section]

\newtheorem{defn}{Definition}[section]

\newtheorem{lemma}{Lemma}[section]
\newtheorem{remark}{Remark}[section]

\begin{document}

\title{Modeling and Identification of Worst-Case Cascading Failures in Power Systems}

\author{Chao Zhai, Hehong Zhang, Gaoxi Xiao and  Tso-Chien Pan \thanks{Chao Zhai, Hehong Zhang, Gaoxi Xiao and Tso-Chien Pan are with Institute of Catastrophe Risk Management, Nanyang Technological University, 50 Nanyang Avenue, Singapore 639798.
They are also with Future Resilient Systems, Singapore-ETH Centre, 1 Create Way, CREATE Tower, Singapore 138602. Hehong Zhang and Gaoxi Xiao are also with School of Electrical and Electronic Engineering, Nanyang Technological University.
Email: EGXXiao@ntu.edu.sg}}

\maketitle

\begin{abstract}
Cascading failures in power systems normally occur as a result of initial disturbance or faults on electrical elements, closely followed by errors of human operators. It remains a great challenge to systematically trace the source of cascading failures in power systems. In this paper, we develop a mathematical model to describe the cascading dynamics of transmission lines in power networks. In particular, the direct current (DC) power flow equation is employed to calculate the transmission power on the branches. By regarding the disturbances on the elements as the control inputs, we formulate the problem of determining the initial disturbances causing the cascading blackout of power grids in the framework of optimal control theory, and the magnitude of disturbances or faults on the selected branch can be obtained by solving the system of algebraic equations. Moreover, an iterative search algorithm is proposed to look for the optimal solution leading to the worst case of cascading failures. Theoretical analysis guarantees the asymptotic convergence of the iterative search algorithm. Finally, numerical simulations are carried out in IEEE 9 Bus System and IEEE 14 Bus System to validate the proposed approach.
\end{abstract}

\section{Introduction}

The stability and secure operation of power grids have a great impact on other interdependent critical infrastructure systems such as energy system, transportation system, finance system and communication system. Nevertheless, contingencies on vulnerable components of power systems and errors of human operators could trigger the chain reactions ending up with the large blackout of power networks. For instance, the North America cascading blackout on August 14, 2003 made 50 million people living without electricity \cite{empg04}. The misoperation of a German operator in November 2006 triggered a chain reaction of power grids and finally caused 15 million Europeans losing access to power \cite{utce07}. Recently, a relay fault near Taj Mahal in India gave rise to a severe cascading blackout on July 31, 2012 affecting 600 million people. Thus, it is vital to identify worst possible attacks or initial disturbances on the critical electrical elements in advance and develop effective protection strategies to alleviate the cascading blackout of power systems.

A cascading blackout of power system is defined as a sequence of component outages that include at least one triggering component outage and subsequent tripping component outages due to the overloading of transmission lines and situational awareness errors of human operators \cite{hine16,dob07,vai12}. Note that a cascading failure does not necessarily lead to
a cascading blackout or load shedding. The existing cascading models basically fall into $3$ categories \cite{hine16}. The first type of models only reveals the topological property and ignores physics of power grids, and thus is unable to accurately describe the cascading evolution of power networks in practice \cite{hine10,yuy16}. The second type of models
focuses on the quasi-steady-state of power systems and computes the power flow on branches by solving the DC or alternating current (AC) power flow equations. The third one resorts to the dynamic modeling in order to investigate the effects of component dynamics on the emergence of cascading failures \cite{cate84,roy94,jia16}. A dynamic model of cascading failure was presented to deal with the interdependencies of different mechanisms \cite{jia16}, which takes into account the transient dynamics of generators and protective relays.

The disturbances on the transmission lines of power grids generally take the form of impedance or admittance changes in existing work \cite{tae16,tarsi70,fol82}. For example, the outage of a transmission line leads to the infinite impedance or zero admittance between two relevant buses.  Linear or nonlinear programming is normally employed to formulate the problem of determining the disruptive disturbances. \cite{tae16} presents two different optimization formulations to analyze the vulnerability of power grids. Specifically, the nonlinear programming is adopted to address the voltage disturbance, and nonlinear bilevel optimization is employed to deal with the power adjustment. The existing work has largely ignored the cascading process of transmission lines when the power system is suffering from disruptive disturbances. Previous optimization formulations are therefore not sufficient to describe the cascading dynamics of transmission lines in practice since the final configuration of power networks strongly depends on the dynamic evolution of transmission lines besides initial conditions.

In this paper we will develop a cascading model of power networks to describe the dynamical evolution of transmission lines. Moreover, the problem of determining the cause of cascading failure is formulated in the framework of optimal control theory by treating the disruptive disturbances of power systems as control inputs in the optimal control system. The proposed approach provides a new insight into tracing disruptive disturbances on vulnerable components of power grids.

The outline of this paper is organized as follows. Section \ref{sec:prob} presents the cascading model of power systems and the optimal control approach. Section \ref{sec:theory} provides theoretical results for the problem of identifying disruptive disturbances, followed by simulations and validation on IEEE 9 and 14 Bus Systems in Section \ref{sec:sim}. Finally, we conclude the paper and discuss future work in Section \ref{sec:con}.

\section{Problem Formulation}\label{sec:prob}

The power system is basically composed of power stations, transformers, power transmission networks, distribution stations and consumers (see Fig. \ref{disps}). In this work, we are interested in identifying disruptive disturbances ($e.g.$, lightning or storm) on transmission lines that trigger the chain reaction and cause the cascading blackout of power grids. The disturbances give rise to the admittance changes of transmission lines, which results in the rebalance of power flow in power grids. The overloading of transmission lines causes certain circuit breakers to sever the corresponding branches and readjust the power network topology. The above process does not pause until the power grid reaches a new steady state and transmission lines are not severed any more. In this section, we will propose a cascading model to describe the cascading process of transmission lines, where the DC power flow equation is solved to obtain the power flow on each branch. More significantly, the mathematical formulation based on optimal control is presented by treating the disruptive disturbances of power grids as the control inputs in the optimal control system.

\begin{figure}
\scalebox{0.05}[0.05]{\includegraphics{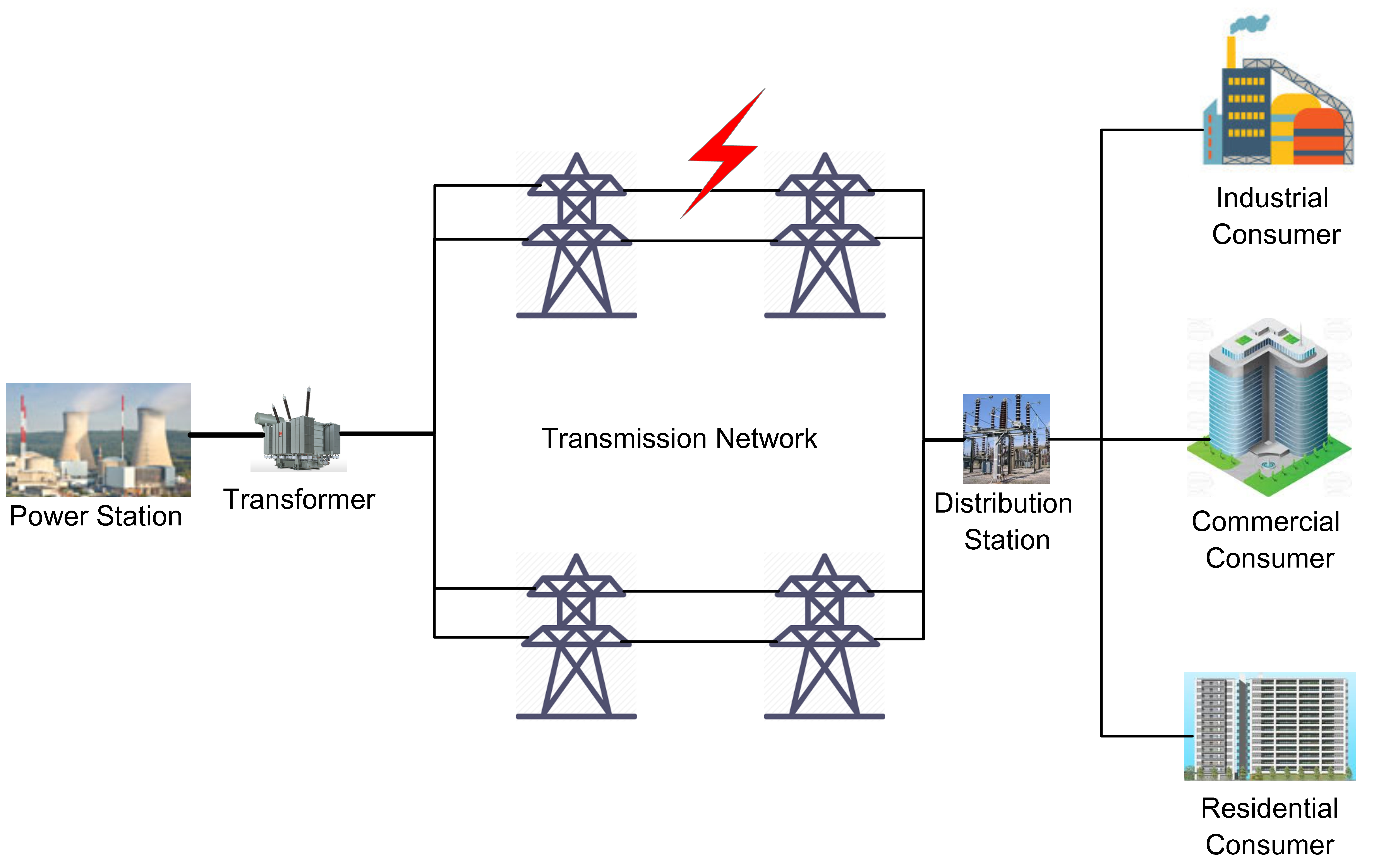}}\centering
\caption{\label{disps} Schematic diagram of power systems suffering from lightning on branches.}
\end{figure}

\subsection{Cascading model}
The cascading model describes the evolution of branch admittance as a result of overloading on transmission lines and the ensuing branch outage. To characterize the connection state of transmission line, we introduce the state function of the transmission line that connects Bus $i$ and Bus $j$ as follows
\begin{equation}\label{thre_fun}
g(P_{ij},c_{ij})=\left\{
                \begin{array}{ll}
                  0, & \hbox{$|P_{ij}|\geq \sqrt{c_{ij}^2+\frac{\pi}{2\sigma}}$;} \\
                  1, & \hbox{$|P_{ij}|\leq\sqrt{c_{ij}^2-\frac{\pi}{2\sigma}}$;} \\
                  \frac{1-\sin\sigma (P_{ij}^2-c_{ij}^2)}{2}, & \hbox{otherwise.}
                \end{array}
              \right.
\end{equation}
where $i$,$j\in I_{n_b}=\{1, 2,...,n_b\}$, $i\neq j$ and $n_b$ is the total number of buses in the power system. $\sigma$ is a tunable positive parameter. $P_{ij}$ refers to the transmitted power on the transmission line that links Bus $i$ and Bus $j$, and $c_{ij}$ denotes its power threshold. The state function $g(P_{ij},c_{ij})$ is differentiable with respect to $P_{ij}$, and it more closely resembles the step function as $\sigma$ increases (see Fig. \ref{thre}). The transmission line is in good condition when $g(P_{ij},c_{ij})=1$, while $g(P_{ij},c_{ij})=0$ implies that the transmission line has been severed by the circuit breaker.
\begin{figure}
\scalebox{0.6}[0.6]{\includegraphics{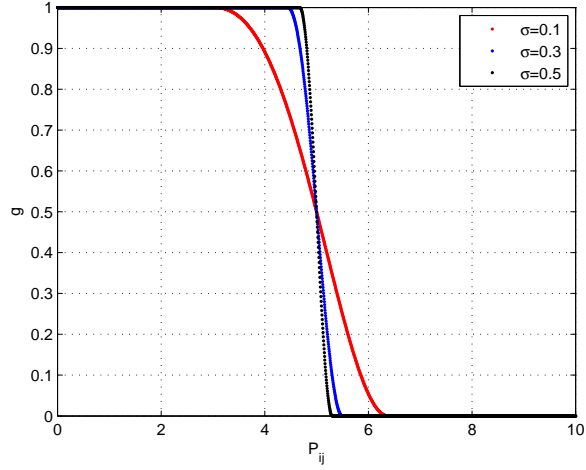}}\centering
\caption{\label{thre} Threshold function $g(P_{ij},c_{ij})$ with $c_{ij}=5$.}
\end{figure}
The cascading model of power systems at the $k$-th step can be presented as
\begin{equation}\label{state_eq}
Y_p^{k+1}=G(P_{ij}^k,c_{ij})\cdot Y_p^{k}+E_{i_k}u_k, \quad k=0,1,2,...m-1
\end{equation}
where  $Y_p^k=(y^k_{p,1},y^k_{p,2},...,y^k_{p,n})^T$ is the admittance vector for the $n$ transmission lines or branches at the $k$-th step, and $u_k=(u_{k,1},u_{k,2},...,u_{k,n})^T$ denotes the control input on transmission
lines. $m$ is the total number of cascading steps in power networks. $G(P_{ij}^k,c_{ij})$ and $E_{i_k}$ are the diagonal matrixes defined as
$$
G(P_{ij}^k,c_{ij})=\left(
              \begin{array}{cccc}
                g(P^k_{i_1j_1},c_{i_1j_1}) & 0 & . & 0 \\
                0 & g(P^k_{i_2j_2},c_{i_2j_2}) & . & 0 \\
                . & . & . & . \\
                0 & 0 & . & g(P^k_{i_nj_n},c_{i_nj_n}) \\
              \end{array}
            \right)
$$
and
$$
E_{i_k}=diag\left(e^T_{i_k}\right)=diag(\underbrace{0,..,0,1}_{i_k},0,...,0)\in R^{n \times n}.
$$
Here $E_{i_k}$ is used to select the $i_k$-th branch to add the control input on. The intuitive interpretation of Cascading Model (\ref{state_eq}) is that the admittance of transmission line becomes zero and remains unchanged after the branch outage and that the control input is added on the selected transmission line to directly change its admittance at the initial step.

\begin{remark}
Without the reclosing operation of circuit breakers, the maximum number of cascading steps $m$ should be less than or equal to $n$, $i.e.$, the total number of transmission lines in the power network. With the reclosing operation of circuit breakers, the cascading model is determined by
$$
Y_p^{k+1}=G(P_{ij}^k,c_{ij})\cdot Y_p^{0}+E_{i_k}u_k \quad k=0,1,2,...m-1
$$
This ensures that the transmission line gets reconnected once its transmission power is less than the specified threshold.
\end{remark}
\subsection{DC power flow equation}
In this work, we focus on the state evolution of transmission lines or mains in power systems and thus compute the DC power flow to deal with the overloading problem.
Specifically, the DC power flow equation is given by
\begin{equation}\label{pfe}
P_i=\sum_{j=1}^{n_b}B_{ij}\theta_{ij}=\sum_{j=1}^{n_b}B_{ij}(\theta_i-\theta_j)
\end{equation}
where $P_i$ and $\theta_i$ refer to the injection power and voltage phase angle of Bus $i$, respectively. $B_{ij}$ represents the mutual susceptance between Bus $i$ and Bus $j$, where $i,j\in I_{n_b}$.  Equation (\ref{pfe}) can be rewritten in matrix form \cite{stot09}
$$
P=B\theta
$$
where
$$
P=\left(
    \begin{array}{c}
      P_1 \\
      P_2 \\
      . \\
      P_{n_b} \\
    \end{array}
  \right),\quad B=\left(
               \begin{array}{cccc}
                 \sum_{i=2}^{n_b}B_{1i} & -B_{12} & . & B_{1n_b} \\
                 -B_{21} & \sum_{i=1,i\neq2}^{n_b}B_{2i} & . & B_{2n_b} \\
                 . & . & . & . \\
                 -B_{n_b1} & -B_{n_b2} & . & \sum_{i=1}^{n_b-1}B_{n_bi} \\
               \end{array}
             \right), \quad \theta=\left(
                    \begin{array}{c}
                      \theta_1 \\
                      \theta_2 \\
                      . \\
                      \theta_{n_b} \\
                    \end{array}
                  \right)
$$
Actually, $B$ is the nodal admittance matrix of power networks while using the DC power flow. The nodal admittance matrix $Y^k_{b}$ at the $k$-th time step can be obtained as
$$
Y^k_{b}=A^Tdiag(Y^k_p)A
$$
where $A$ denotes the branch-bus incidence matrix \cite{stag68}. Therefore, the matrix $B$ at the $k$-th time step of cascading failure can be calculated as
$$
B^k=Y^k_{b}=A^Tdiag(Y^k_p)A, \quad Y^k_p=(y^k_{p,1},y^k_{p,2},...,y^k_{p,n})^T, \quad y^k_{p,i}=-\frac{1}{Im(z^k_{p,i})}, \quad  i\in I_n=\{1,2,...,n\}
$$
where $z^k_{p,i}$ denotes the impedance of the $i$-th branch at the $k$-th time step. Then the DC power flow equation at the $k$-th time step is given by
\begin{equation}\label{dc_pfe}
    P^k=B^k\theta^k=Y_b^k\theta^k
\end{equation}
where $P^k=(P^k_1,P^k_2,...,P^k_{n_b})^T$ and $\theta^k=(\theta^k_1,\theta^k_2,...,\theta^k_{n_b})^T$.
During the cascading blackout, the power network may be divided into several subnetworks ($i.e.$, islands), which can be identified by analyzing the nodal admittance matrix $Y^k_{b}$. Suppose $Y^k_{b}$ is composed of $q$ isolated components or subnetworks denoted by $\mathrm{S}_i, i\in I_q=\{1,2,...,q\}$
and each subnetwork $\mathrm{S}_i$ includes $k_i$ buses, $i.e.$, $\mathrm{S}_i=\{i_1,i_2,...,i_{k_i}\}$, where $i_1$, $i_2$,..., $i_{k_i}$ denote the bus identity (ID) numbers and $\sum_{i=1}^{q}k_i=n_b$. Notice that Bus $i_1$ in Subnetwork $\mathrm{S}_i$ is designated as the reference bus, which is normally a generator bus in practice. The nodal admittance matrix of the $i$-th subnetwork can be computed as
$$
Y^k_{b, i}=\left(
  \begin{array}{c}
    e_{i_1}^T \\
    e_{i_2}^T \\
    . \\
    e_{i_{k_i}}^T \\
  \end{array}
\right)Y^k_{b}\left(e_{i_1}, e_{i_2} ,..., e_{i_{k_i}}\right), \quad i\in I_q
$$
For simplicity, we introduce two operators $*$ and $-1^*$ to facilitate the analytical expression and theoretical analysis of solving the DC power flow equation.
\begin{defn}\label{def}
Given the nodal admittance matrix $Y^k_b$, the operators $*$ and $-1^*$ are defined by
$$
\left(Y^k_{b}\right)^*=\sum_{i=1}^{q}\left(e_{i_1}, e_{i_2} ,..., e_{i_{k_i}}\right)\left(
           \begin{array}{c|c}
             0 & 0_{k_i-1}^T \\\hline
             0_{k_i-1} & I_{k_i-1} \\
           \end{array}
         \right)Y^k_{b,i}\left(
           \begin{array}{c|c}
             0 & 0_{k_i-1}^T \\\hline
             0_{k_i-1} & I_{k_i-1} \\
           \end{array}
         \right)\left(
  \begin{array}{c}
    e_{i_1}^T \\
    e_{i_2}^T \\
    . \\
    e_{i_{k_i}}^T \\
  \end{array}
\right)
$$
and
$$
\left(Y^k_{b}\right)^{-1^*}=\sum_{i=1}^{q}\left(e_{i_1}, e_{i_2} ,..., e_{i_{k_i}}\right)\left(
           \begin{array}{c}
             0^T_{k_i-1} \\
             I_{k_i-1} \\
           \end{array}
         \right)\left[\left(
                                                                                                 \begin{array}{cc}
                                                                                                   0_{k_i-1} & I_{k_i-1} \\
                                                                                                 \end{array}
                                                                                               \right)
Y^k_{b,i}\left(
           \begin{array}{c}
             0^T_{k_i-1} \\
             I_{k_i-1} \\
           \end{array}
         \right)
\right]^{-1}\left(
                                                                                                 \begin{array}{cc}
                                                                                                   0_{k_i-1} & I_{k_i-1} \\
                                                                                                 \end{array}
                                                                                               \right)\left(
  \begin{array}{c}
    e_{i_1}^T \\
    e_{i_2}^T \\
    . \\
    e_{i_{k_i}}^T \\
  \end{array}
\right),
$$
respectively, where
$$
I_{k_i-1}=\left(
            \begin{array}{cccc}
              1 & 0 & . & 0 \\
              0 & 1 & . & 0 \\
              . & . & . & 0 \\
              0 & 0 & . & 1 \\
            \end{array}
          \right)\in R^{(k_i-1)\times(k_i-1)}, \quad
0_{k_i-1}=\left(
                                  \begin{array}{c}
                                    0 \\
                                    0 \\
                                    . \\
                                    0 \\
                                  \end{array}
                                \right)\in R^{k_i-1}
$$
\end{defn}
\begin{remark}
The power network represented by the nodal admittance matrix $Y_b^k$ can be decomposed into $q$ isolated subnetworks, and each subnetwork is described by a submatrix $Y_{b,i}^k,i\in I_q$. The operators $*$ and $-1^*$ replace all the elements in the $1$-st row and the $1$-st column of $Y_{b,i}^k$ with $0$. Moreover, the operator $-1^*$ also replaces the remaining part of $Y_{b,i}^k$ with its inverse matrix. According to algebraic graph theory, the rank of nodal admittance matrix $Y^k_{b,i}$ is $k_i-1$ since each subnetwork $S_i, i\in \{1,2,...,q\}$ is connected \cite{god13}. Thus, it is guaranteed that the matrix
$$
\left(
     \begin{array}{cc}
       0_{k_i-1} & I_{k_i-1} \\
     \end{array}
   \right)
Y^k_{b,i}\left(
           \begin{array}{c}
             0^T_{k_i-1} \\
             I_{k_i-1} \\
           \end{array}
         \right)
$$
has full rank $k_i-1$ and thus it is invertible.
\end{remark}

\subsection{Optimization formulation}
The cascading dynamics of power system is composed of the cascading model defined by Equation (\ref{state_eq}) and DC power flow equation described by Equation (\ref{dc_pfe}), and these two components are coupled with each other, which characterizes the cascading blackout of power grids after suffering from disruptive disturbances. The optimal control algorithm allows us to obtain the disruptive disturbances by treating the disturbances as the control inputs of the optimal control system (see Fig. \ref{flow}). Specifically, the cascading model describes the outage of overloading branches and updates the admittance on transmission lines with the latest power flow, which is provided by the DC power flow equation. Meanwhile, the DC power flow equation is solved with the up-to-date admittance of branches from the cascading model. The above two processes occur iteratively in describing the evolution of admittance and transmission power on transmission lines. Moreover, the cascading dynamics of power system exactly functions as the state equation of optimal control system. In this way, the optimal control algorithm allows us to gain the control input or disruptive disturbance that triggers the chain reaction of the proposed cascading model.

The identification of disruptive disturbances in power systems can be formulated as the following optimal control problem.
\begin{figure}
\scalebox{0.06}[0.06]{\includegraphics{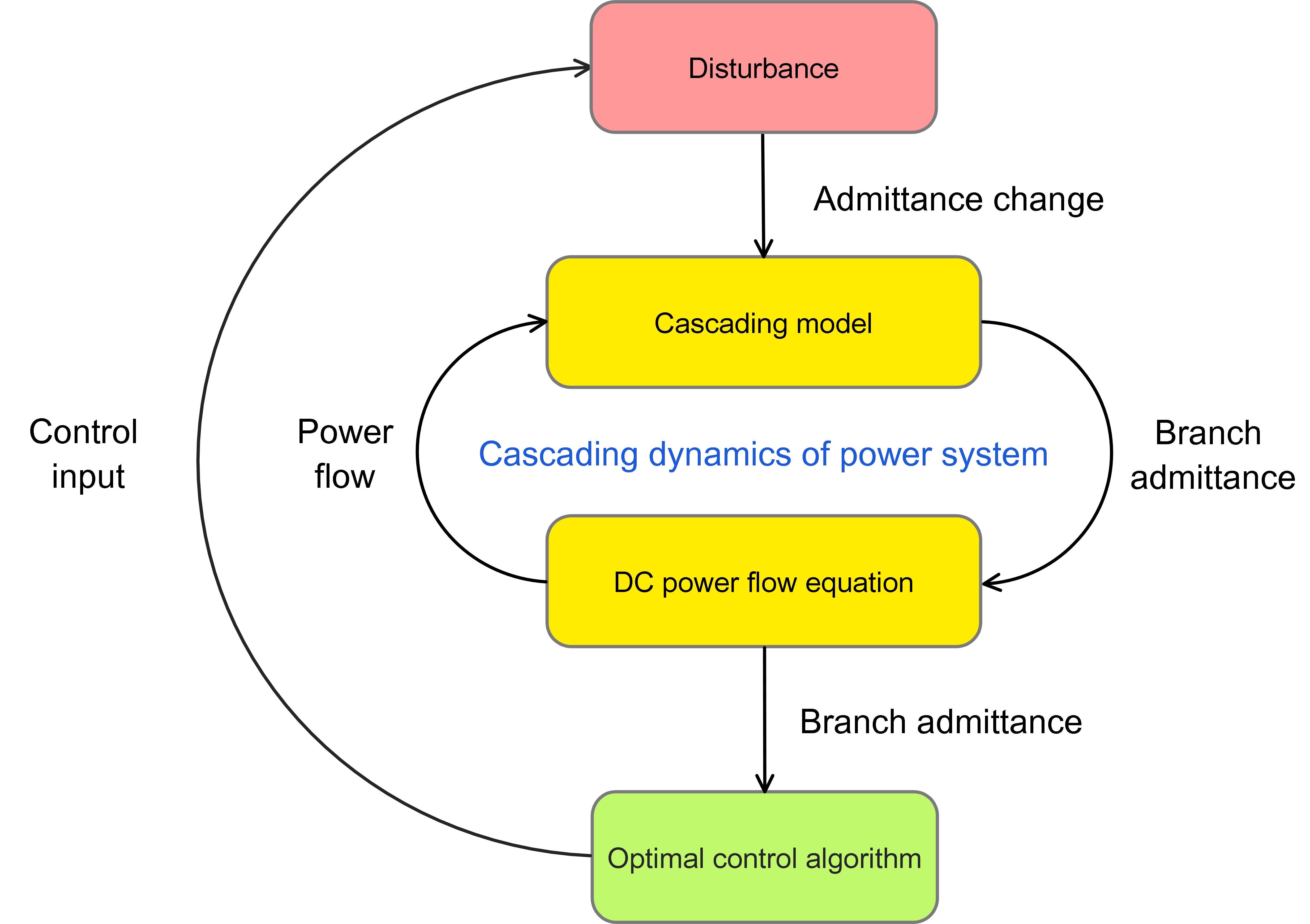}}\centering
\caption{\label{flow} Optimal control approach to identifying disruptive disturbances.}
\end{figure}

\begin{equation}\label{cost}
\min_{u_k}J(Y_p^m,u_k)
\end{equation}
with the cost function
\begin{equation}\label{cost_fun}
J(Y_p^m,u_k)=\mathrm{T}(Y_p^m)+\epsilon\sum_{k=0}^{m-1}\frac{\|u_k\|^2}{\max\{0,\iota-k\}}
\end{equation}
where $\mathbf{1}_n=(1,1,...,1)^T\in R^n$ and $\epsilon$ is a positive weight. $\|\cdot\|$ represents the 2-norm. As mentioned before, the state equation of the optimal control system consists of Equations (\ref{state_eq}) and (\ref{dc_pfe}). The above cost function includes two terms. Specifically, the first term $\mathrm{T}(Y_p^m)$ is differentiable with respect to $Y_p^m$, and it quantifies the power state or connectivity of power networks at the final step of cascading blackout, and the second term characterizes the control energy at the first $\iota$ time steps with the constraint $1\leq \iota\leq (m-1)$. In practice, $\mathrm{T}(Y_p^m)$ is designed according to the specific concerns about the worst-case scenario in power systems. In particular, the parameter $\epsilon$ is set small enough so that the first term dominates in the cost function. The objective is to minimize $\mathrm{T}(Y_p^m)$ by adding the appropriate control input $u_k$ on the branches of power systems at the given time step.

\begin{remark}
In practice, the disruptive disturbances of power systems come from contingencies such as lightning and situational awareness errors of human operators, etc.
\end{remark}

\section{Theoretical Analysis}\label{sec:theory}
In this section, we will present some theoretical results on the proposed optimal control algorithm. First of all, the properties of operators $*$ and $-1^*$ are given by the following lemmas.

\begin{lemma}\label{lem_def}
For the nodal admittance matrix $Y^k_b\in R^{n_b\times n_b}$, the equations
\begin{equation}\label{prod}
\begin{split}
    \left(Y^k_{b}\right)^*\left(Y^k_{b}\right)^{-1^*}=\left(Y^k_{b}\right)^{-1^*}\left(Y^k_{b}\right)^{*}&=\sum_{i=1}^{q}\left(e_{i_1}, e_{i_2} ,..., e_{i_{k_i}}\right)diag(0,1^T_{k_i-1})\left(
  \begin{array}{c}
    e_{i_1}^T \\
    e_{i_2}^T \\
    . \\
    e_{i_{k_i}}^T \\
  \end{array}
\right)
\end{split}
\end{equation}
hold.
\end{lemma}

\begin{proof}
It follows from Definition \ref{def} that
\begin{equation*}
\begin{split}
\left(Y^k_{b}\right)^*\left(Y^k_{b}\right)^{-1^*}
         &=\sum_{i=1}^{q}\left(e_{i_1}, e_{i_2} ,..., e_{i_{k_i}}\right)\left(
           \begin{array}{c|c}
             0 & 0_{k_i-1}^T \\\hline
             0_{k_i-1} & I_{k_i-1} \\
           \end{array}
         \right)Y^k_{b,i}\left(
           \begin{array}{c|c}
             0 & 0_{k_i-1}^T \\\hline
             0_{k_i-1} & I_{k_i-1} \\
           \end{array}
         \right)\\
         &~~~\cdot\left(
           \begin{array}{c}
             0^T_{k_i-1} \\
             I_{k_i-1} \\
           \end{array}
         \right)\left[\left(
                     \begin{array}{cc}
                       0_{k_i-1} & I_{k_i-1} \\
                     \end{array}
                   \right)
            Y^k_{b,i}\left(
                       \begin{array}{c}
                         0^T_{k_i-1} \\
                         I_{k_i-1} \\
                       \end{array}
                     \right)
            \right]^{-1}\left(
                         \begin{array}{cc}
                           0_{k_i-1} & I_{k_i-1} \\
                         \end{array}
                       \right)\left(
              \begin{array}{c}
                e_{i_1}^T \\
                e_{i_2}^T \\
                . \\
                e_{i_{k_i}}^T \\
              \end{array}
            \right)\\
         &=\sum_{i=1}^{q}\left(e_{i_1}, e_{i_2} ,..., e_{i_{k_i}}\right)\left(
           \begin{array}{c}
             0^T_{k_i-1} \\
             I_{k_i-1} \\
           \end{array}
         \right)\left[\left(
                     \begin{array}{cc}
                       0_{k_i-1} & I_{k_i-1} \\
                     \end{array}
                   \right)
            Y^k_{b,i}\left(
                       \begin{array}{c}
                         0^T_{k_i-1} \\
                         I_{k_i-1} \\
                       \end{array}
                     \right)
            \right]\left(
                         \begin{array}{cc}
                           0_{k_i-1} & I_{k_i-1} \\
                         \end{array}
                       \right)\\
         &~~~\cdot\left(
           \begin{array}{c}
             0^T_{k_i-1} \\
             I_{k_i-1} \\
           \end{array}
         \right)\left[\left(
                     \begin{array}{cc}
                       0_{k_i-1} & I_{k_i-1} \\
                     \end{array}
                   \right)
            Y^k_{b,i}\left(
                       \begin{array}{c}
                         0^T_{k_i-1} \\
                         I_{k_i-1} \\
                       \end{array}
                     \right)
            \right]^{-1}\left(
                         \begin{array}{cc}
                           0_{k_i-1} & I_{k_i-1} \\
                         \end{array}
                       \right)\left(
              \begin{array}{c}
                e_{i_1}^T \\
                e_{i_2}^T \\
                . \\
                e_{i_{k_i}}^T \\
              \end{array}
            \right)\\.
\end{split}
\end{equation*}
Moreover, it follows from
\begin{equation*}
\begin{split}
&\left[\left(
                     \begin{array}{cc}
                       0_{k_i-1} & I_{k_i-1} \\
                     \end{array}
                   \right)
            Y^k_{b,i}\left(
                       \begin{array}{c}
                         0^T_{k_i-1} \\
                         I_{k_i-1} \\
                       \end{array}
                     \right)
            \right]\left(
                         \begin{array}{cc}
                           0_{k_i-1} & I_{k_i-1} \\
                         \end{array}
                       \right)\left(
           \begin{array}{c}
             0^T_{k_i-1} \\
             I_{k_i-1} \\
           \end{array}
         \right)\left[\left(
                     \begin{array}{cc}
                       0_{k_i-1} & I_{k_i-1} \\
                     \end{array}
                   \right)
            Y^k_{b,i}\left(
                       \begin{array}{c}
                         0^T_{k_i-1} \\
                         I_{k_i-1} \\
                       \end{array}
                     \right)
            \right]^{-1}\\
&=\left[\left(
                     \begin{array}{cc}
                       0_{k_i-1} & I_{k_i-1} \\
                     \end{array}
                   \right)
            Y^k_{b,i}\left(
                       \begin{array}{c}
                         0^T_{k_i-1} \\
                         I_{k_i-1} \\
                       \end{array}
                     \right)
            \right]I_{k_i-1}\left[\left(
                     \begin{array}{cc}
                       0_{k_i-1} & I_{k_i-1} \\
                     \end{array}
                   \right)
            Y^k_{b,i}\left(
                       \begin{array}{c}
                         0^T_{k_i-1} \\
                         I_{k_i-1} \\
                       \end{array}
                     \right)
            \right]^{-1}\\
&=I_{k_i-1}
\end{split}
\end{equation*}
that
$$
\left(Y^k_{b}\right)^{*}\left(Y^k_{b}\right)^{-1^*}=\sum_{i=1}^{q}\left(e_{i_1}, e_{i_2} ,..., e_{i_{k_i}}\right)diag(0,1^T_{k_i-1})\left(
  \begin{array}{c}
    e_{i_1}^T \\
    e_{i_2}^T \\
    . \\
    e_{i_{k_i}}^T \\
  \end{array}
\right).
$$
Likewise, we can prove
$$
\left(Y^k_{b}\right)^{-1^*}\left(Y^k_{b}\right)^{*}=\sum_{i=1}^{q}\left(e_{i_1}, e_{i_2} ,..., e_{i_{k_i}}\right)diag(0,1^T_{k_i-1})\left(
  \begin{array}{c}
    e_{i_1}^T \\
    e_{i_2}^T \\
    . \\
    e_{i_{k_i}}^T \\
  \end{array}
\right).
$$
\end{proof}
Lemma \ref{lem_def} indicates that the two operators $*$ and $-1^*$ are commutative for the same square matrix. Given the injection power for each bus $P^k=(P^k_1,P^k_2,...,P^k_{n_b})^T$ at the $k$-th time step, the quantitative relationship between $Y^k_p$ and power flow on each branch is presented as follows.
\begin{lemma}\label{lem_pij}
$$
P^k_{ij}=e_i^TA^T diag(Y^k_p)Ae_j(e_i-e_j)^T(A^T diag(Y^k_p)A)^{-1^*}P^k, \quad i,j\in I_{n_b}
$$
where $e_i=(\underbrace{0,...,0,1}_{i-th},0,...0)^T\in R^n$.
\end{lemma}

\begin{proof}
It follows from the solution to DC power flow equation $\theta^k=(B^k)^{-1^*}P^k$ and $B^k=A^T diag(Y^k_p)A$ that
\begin{equation*}
\begin{split}
P^k_{ij}=B^k_{ij}(\theta^k_i-\theta^k_j)&=e_i^TB^ke_j(e_i-e_j)^T\theta^k \\
&=e_i^TB^ke_j(e_i-e_j)^T(B^k)^{-1^*}P^k \\
&=e_i^TA^T diag(Y^k_p)Ae_j(e_i-e_j)^T(A^T diag(Y^k_p)A)^{-1^*}P^k.
\end{split}
\end{equation*}
\end{proof}

Similar to the matrix inversion, the operators $*$ and $-1^*$ satisfy the following equation in terms of the derivative operation.
\begin{lemma}\label{lem_der}
$$
\frac{\partial (A^T diag(Y^k_p)A)^{-1^*}}{\partial y^k_{p,i}}=-(A^T diag(Y^k_p)A)^{-1^*}(A^Tdiag(e_i)A)^*(A^T diag(Y^k_p)A)^{-1^*}, \quad i\in I_n.
$$
\end{lemma}

\begin{proof}
Lemma \ref{lem_def} allows to obtain
$$
(A^T diag(Y^k_p)A)^*\cdot(A^T diag(Y^k_p)A)^{-1^*}=\sum_{i=1}^{q}\left(e_{i_1}, e_{i_2} ,..., e_{i_{k_i}}\right)diag(0,1^T_{k_i-1})\left(
  \begin{array}{c}
    e_{i_1}^T \\
    e_{i_2}^T \\
    . \\
    e_{i_{k_i}}^T \\
  \end{array}
\right).
$$
Since the derivative of the constant is $0$, we have
\begin{equation*}
\begin{split}
&~~~~\frac{\partial [(A^T diag(Y^k_p)A)^*\cdot(A^T diag(Y^k_p)A)^{-1^*}]}{\partial y^k_{p,i}}\\
&=\frac{\partial (A^T diag(Y^k_p)A)^*}{\partial y^k_{p,i}}\cdot(A^T diag(Y^k_p)A)^{-1^*}+(A^T diag(Y^k_p)A)^* \cdot \frac{\partial (A^T diag(Y^k_p)A)^{-1^*}}{\partial y^k_{p,i}}\\
&=\frac{\partial}{\partial y^k_{p,i}}\sum_{i=1}^{q}\left(e_{i_1}, e_{i_2} ,..., e_{i_{k_i}}\right)diag(0,1^T_{k_i-1})\left(
  \begin{array}{c}
    e_{i_1}^T \\
    e_{i_2}^T \\
    . \\
    e_{i_{k_i}}^T \\
  \end{array}
\right)=0_{n_b\times n_b}.
\end{split}
\end{equation*}
Then it follows from
\begin{equation*}
\begin{split}
&(A^T diag(Y^k_p)A)^{-1^*}[\frac{\partial (A^T diag(Y^k_p)A)^*}{\partial y^k_{p,i}}\cdot(A^T diag(Y^k_p)A)^{-1^*}+(A^T diag(Y^k_p)A)^* \cdot \frac{\partial (A^T diag(Y^k_p)A)^{-1^*}}{\partial y^k_{p,i}}]\\
&=(A^T diag(Y^k_p)A)^{-1^*}\cdot\frac{\partial (A^T diag(Y^k_p)A)^*}{\partial y^k_{p,i}}\cdot(A^T diag(Y^k_p)A)^{-1^*}+diag(0,1^T_{n-1})\cdot \frac{\partial (A^T diag(Y^k_p)A)^{-1^*}}{\partial y^k_{p,i}}\\
&=(A^T diag(Y^k_p)A)^{-1^*}\cdot\frac{\partial (A^T diag(Y^k_p)A)^*}{\partial y^k_{p,i}}\cdot(A^T diag(Y^k_p)A)^{-1^*}+\frac{\partial (A^T diag(Y^k_p)A)^{-1^*}}{\partial y^k_{p,i}}=0_{n_b\times n_b}
\end{split}
\end{equation*}
that
\begin{equation*}
\begin{split}
\frac{\partial (A^T diag(Y^k_p)A)^{-1^*}}{\partial y^k_{p,i}}&=-(A^T diag(Y^k_p)A)^{-1^*}\frac{\partial(A^Tdiag(Y^k_p)A)^*}{\partial y^k_{p,i}}(A^T diag(Y^k_p)A)^{-1^*}\\
&=-(A^T diag(Y^k_p)A)^{-1^*}(A^Tdiag(\frac{\partial Y^k_P}{\partial y^k_{p,i}})A)^*(A^T diag(Y^k_p)A)^{-1^*}\\
&=-(A^T diag(Y^k_p)A)^{-1^*}(A^Tdiag(e_i)A)^*(A^T diag(Y^k_p)A)^{-1^*}.
\end{split}
\end{equation*}
\end{proof}
Next, we present theoretical results relevant to optimal control problem (\ref{cost}). For the discrete time nonlinear system, optimal control theory provides the necessary conditions for deriving the control input to minimize the given cost function.
\begin{theorem}\label{disopt}
For the discrete time optimal control problem
$$
\min_{u_k}J(x_k,u_k)
$$
with the state equation
$$
x_{k+1}=f(x_k,u_k), \quad k=0,1,...,m-1
$$
and the cost function
$$
J(x_k,u_k)=\phi(x_m)+\sum_{k=0}^{m-1}L(x_k,u_k),
$$
the necessary conditions for the optimal control input $u_k^*$ are given as follows
\begin{enumerate}
  \item $x_{k+1}=f(x_k,u_k)$
  \item $\lambda_k=(\frac{\partial f}{\partial x_k})^T\lambda_{k+1}+\frac{\partial L}{\partial x_k}$
  \item $(\frac{\partial f}{\partial u_k})^T\lambda_{k+1}+\frac{\partial L}{\partial u_k}=0$
  \item $\lambda_m=\frac{\partial \Phi}{\partial x_m}$
\end{enumerate}
\end{theorem}

\begin{proof}
It is a special case ($i.e.$, time invariant case) of the optimal control for the time-varying discrete time nonlinear system in \cite{fran95}. Hence the proof is omitted.
\end{proof}
By applying Theorem \ref{disopt} to the optimal control problem (\ref{cost}), we obtain the necessary conditions for identifying the disruptive disturbance of power systems with the cascading model (\ref{state_eq}) and the DC power flow equation (\ref{dc_pfe}).
\begin{theorem}\label{sysAE}
The necessary condition for the optimal control problem (\ref{cost}) corresponds to the solution of the following system of algebraic equations.
\begin{equation}\label{con_sys}
Y_p^{k+1}-G(P^k_{ij},c_{ij})Y_p^k+\frac{\max\{0,\iota-k\}}{2\epsilon}E_{i_k}\prod_{s=0}^{m-k-2}\frac{\partial Y_p^{m-s}}{\partial Y_p^{m-s-1}}\cdot\frac{\partial\mathrm{T}(Y_p^m)}{\partial Y_p^m}=\mathbf{0}_n, \quad k=0,1,...,m-1
\end{equation}
and the optimal control input is given by
\begin{equation}\label{con_input}
u_k=-\frac{\max\{0,\iota-k\}}{2\epsilon}E_{i_k}\prod_{s=0}^{m-k-2}\frac{\partial Y_p^{m-s}}{\partial Y_p^{m-s-1}}\cdot \frac{\partial \mathrm{T}(Y_p^m)}{\partial Y_p^m}, \quad k=0,1,...,m-1
\end{equation}
\end{theorem}

\begin{proof}
See Appendix.
\end{proof}

It is necessary to winnow the solutions to Equation (\ref{con_sys}), since they just satisfy necessary conditions for optimal control problem (\ref{cost}). Thus, we introduce a search algorithm to explore the optimal control input or initial disturbances. Table \ref{ISA} presents the implementation process of the Iterative Search Algorithm (ISA) in details. First of all, we set the maximum iterative steps $i_{max}$ of the ISA  and the initial value of cost function $J^*$, which is a sufficiently large number $J_{\max}$ and is larger than the maximum value of the cost function. The solution to the system of algebraic equation (\ref{con_sys}) allows us to obtain the control input $u^i$ from (\ref{con_input}). Then we compute the cost function $J^i$ from (\ref{cost}) by adding the control input $u^i$ in power systems. Then $J^*$ and $u^*$ are replaced with $J^i$ and $u^i$ if $J^i$ is less than $J^*$. Afterwards, the algorithm goes to the next iteration and starts solving the system of algebraic equation (\ref{con_sys}) once again.

\begin{table}
 \caption{\label{ISA} Iterative Search Algorithm.}
 \begin{center}
 \begin{tabular}{lcl} \hline
  1: Set the maximum number of steps $i_{\max}$, $i=0$ and $J^*=J_{\max}$ \\
  2: \textbf{while} ($i<=i_{\max}$) \\
  3: ~~~~~~~Solve the system of algebraic equation (\ref{con_sys})  \\
  4: ~~~~~~~Compute the control input $u^i$ from (\ref{con_input})  \\
  5: ~~~~~~~Validate the control input $u^i$ in (\ref{state_eq})\\
  6: ~~~~~~~Compute the resulting cost function $J^i$ from (\ref{cost_fun}) \\
  7: ~~~~~~~\textbf{if} ($J^i<J^*$)  \\
  8: ~~~~~~~~~~~Set $u^*=u^i$ and $J^*=J^i$ \\
  9: ~~~~~~~\textbf{end if}  \\
 10: ~~~~~Set $i=i+1$ \\
 11: \textbf{end while} \\ \hline
 \end{tabular}
 \end{center}
\end{table}

Regarding the Iterative Search Algorithm in Table \ref{ISA}, we have the following theoretical result.
\begin{theorem}
The Iterative Search Algorithm in Table \ref{ISA} ensures that the cost function $J^*$ and control input $u^*$ converge to the optima as the iteration steps $i_{\max}$ go to infinity.
\end{theorem}

\begin{proof}
The ISA in Table \ref{ISA} indicates that the cost function $J^*$ decreases monotonically  as time goes. Considering that $J^*$ is the lower bounded ($i.e.$, $J^*\geq0$), it can be proved that $J^*$ converges to the infimum according to monotone convergence theorem in real analysis \cite{yeh06}. For each iteration, the system of algebraic equation (\ref{con_sys}) is solved with a random initial condition. As a result, the cost function $J^*$ and control input $u^*$ converge to the optima as the iteration steps $i_{\max}$ go to infinity.
\end{proof}

\section{Simulation and Validation}\label{sec:sim}

In this section, we implement the proposed Disturbance Identification Algorithm to search for the disruptive disturbances added on selected branches in IEEE 9 Bus System and IEEE 14 Bus System. The numerical results on disruptive disturbances
are validated by disturbing the selected branch with the computed magnitude of disturbance in the corresponding IEEE Bus Systems. To sever as many branches as possible, we define the terminal constraint in cost function (\ref{cost_fun}) as follows
$$
T(Y_p^m)=\frac{1}{2}\|Y_p^m\|^2
$$
and derive its partial derivative with respect to $Y_p^m$
\begin{equation}\label{part_term}
\frac{\partial\mathrm{T}(Y_p^m)}{\partial Y_p^m}=Y_p^m
\end{equation}
By substituting (\ref{part_term}) into (\ref{con_sys}), we obtain the desired system of algebraic equations.
\begin{equation}\label{sys_yp}
Y_p^{k+1}-G(P^k_{ij},c_{ij})Y_p^k+\frac{\max\{0,\iota-k\}}{2\epsilon}E_{i_k}\prod_{s=0}^{m-k-2}\frac{\partial Y_p^{m-s}}{\partial Y_p^{m-s-1}}Y_p^m=\mathbf{0}_n, \quad k=0,1,...,m-1.
\end{equation}

\subsection{IEEE 9 Bus System}

The parameter settings of IEEE 9 Bus System (see Fig. \ref{ieee9b}) are presented in Table \ref{tab:branch9} and Table \ref{tab:bus9} \cite{zim11}. It is worth noting that R represents the reference bus (slack bus). G refers to the generator bus and L stands for the load bus in Table \ref{tab:bus9}. Per unit values are adopted with the base value $100$MVA. Other parameters for the dynamic model of power system are given as $\sigma=5\times10^4$, $\iota=1$, $\epsilon=10^{-4}$, $i_{\max}=10$, $J_{\max}=10^6$ and $m=9$. The solver ``fsolve" in Matlab is employed to solve the system of algebraic equations (\ref{sys_yp}).

\begin{table}[h!]
  \centering
  \caption{IEEE 9 Bus System-Branch data.}
  \label{tab:branch9}
  \begin{tabular}{ccccc}
    \toprule
    Branch number & Source bus & Sink bus & Reactance & Power threshold \\
    \midrule
      1 & 1 & 4 & 0.058 & 1.0 \\
      2 & 2 & 7 & 0.092 & 1.8 \\
      3 & 3 & 9 & 0.170 & 1.0 \\
      4 & 4 & 5 & 0.059 & 0.5 \\
      5 & 4 & 6 & 0.101 & 0.5 \\
      6 & 7 & 5 & 0.072 & 1.0 \\
      7 & 7 & 8 & 0.063 & 1.0 \\
      8 & 9 & 6 & 0.161 & 1.0 \\
      9 & 9 & 8 & 0.085 & 1.0 \\
    \bottomrule
  \end{tabular}
\end{table}

\begin{table}[h!]
  \centering
  \caption{IEEE 9 Bus System-Bus data.}
  \label{tab:bus9}
  \begin{tabular}{ccc}
    \toprule
    Bus number & Bus type & Power injection \\
    \midrule
      1 & R & 0.71  \\
      2 & G & 1.63  \\
      3 & G & 0.85  \\
      4 & L & 0  \\
      5 & L & -1.25  \\
      6 & L & -0.9  \\
      7 & L & 0   \\
      8 & L & -1  \\
      9 & L & 0   \\
    \bottomrule
  \end{tabular}
\end{table}

\begin{figure}
\scalebox{0.9}[0.9]{\includegraphics{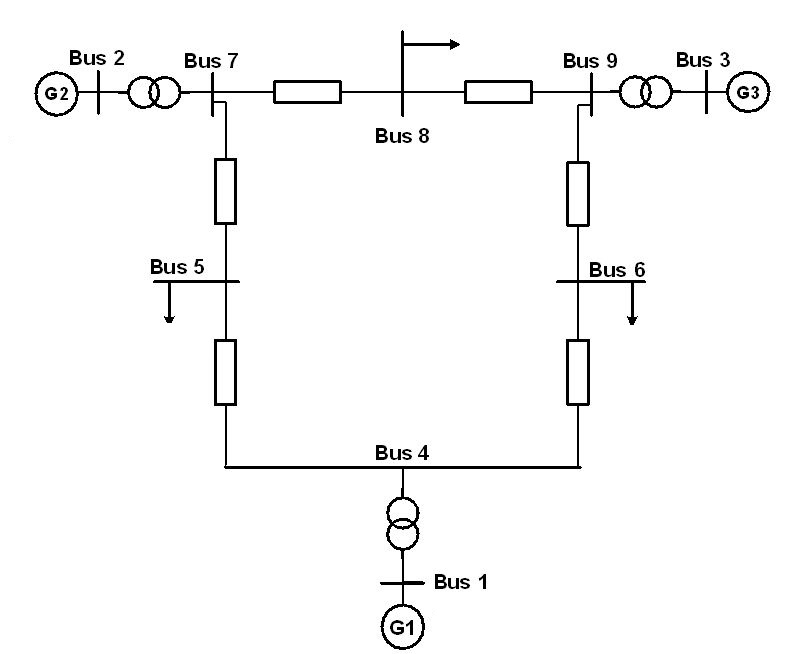}}\centering
\caption{\label{ieee9b} IEEE 9 Bus System.}
\end{figure}

Figure \ref{id9} shows the computed disturbance and corresponding cost for each branch by the ISA in Table \ref{ISA}. It is observed that the disturbance on Branch 2 results in the least cost, which indicates the most outage branches in the final step. In particular, Fig. \ref{ju} presents the time evolution while applying the ISA to search for the desired disturbance
or optimal control input on Branch 2 in 10 rounds. After 7 rounds, the cost function is lowered greatly to the bottom and keeps invariant afterwards. Correspondingly, the computed control input converges to $10.87$, which exactly severs Branch 2.

\begin{figure}\centering
 {\includegraphics[width=0.41\textwidth]{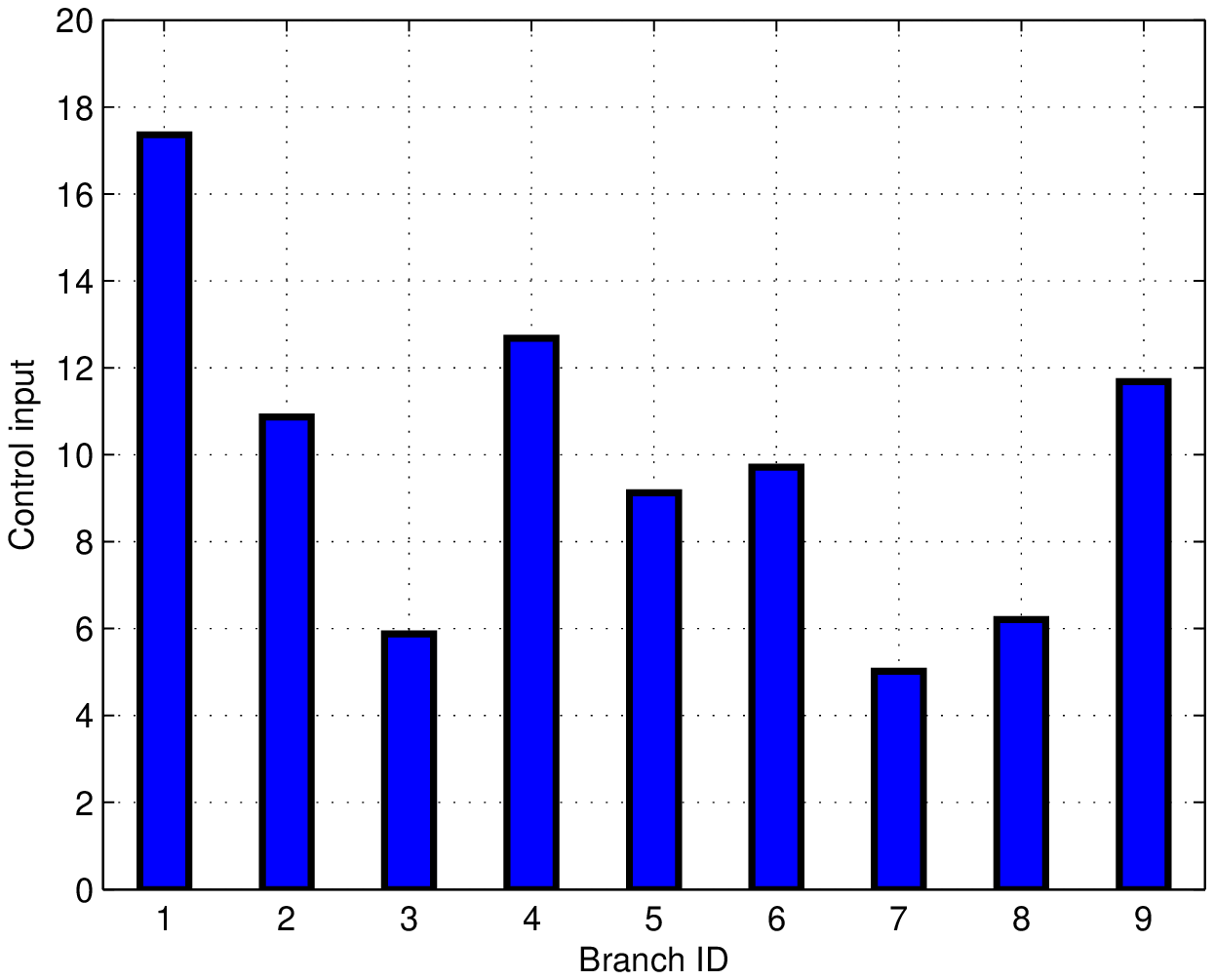}}
 {\includegraphics[width=0.41\textwidth]{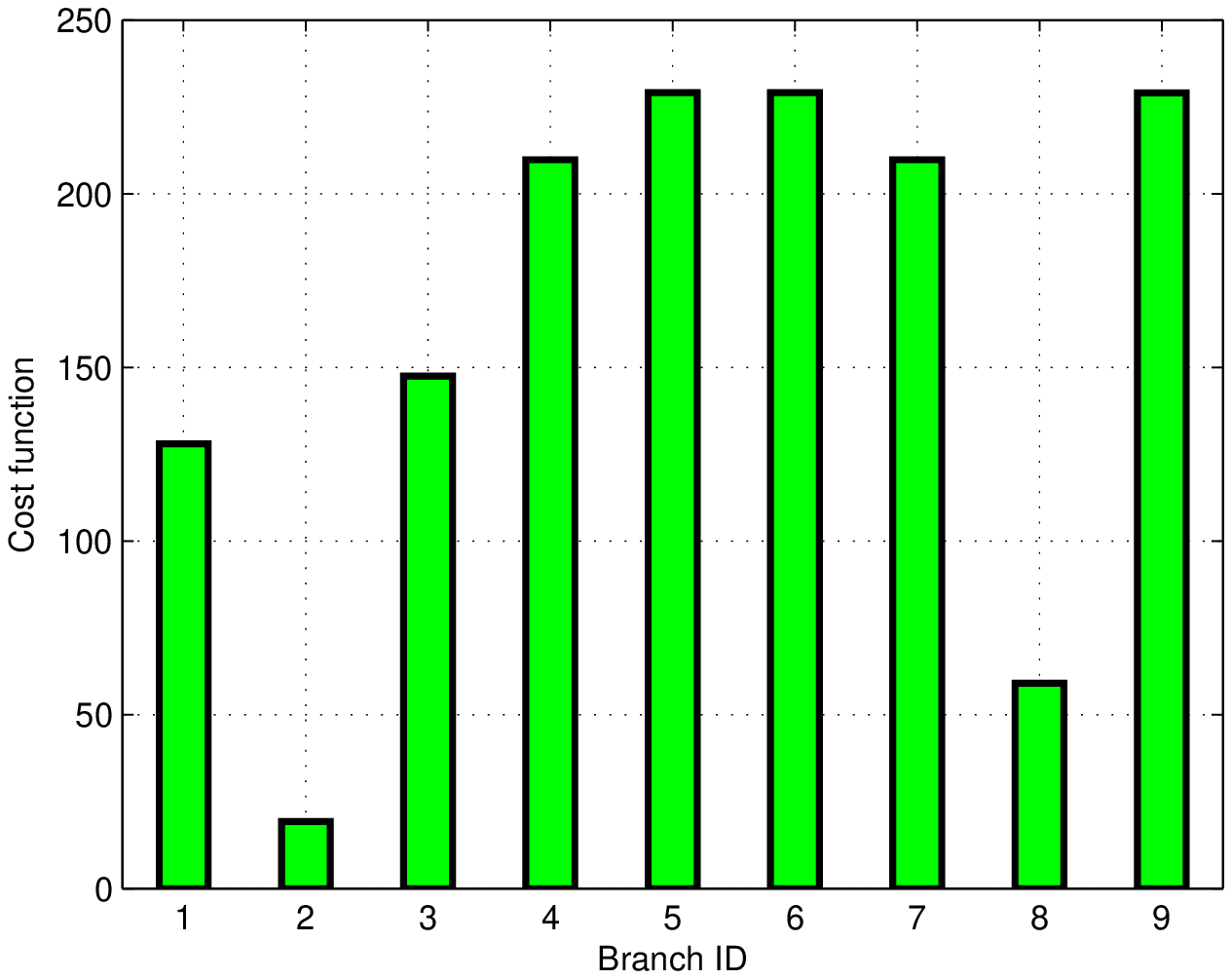}} \\
 \caption{\label{id9} Control input and the resulted cost on each branch in IEEE 9 Bus System.}
\end{figure}

\begin{figure}\centering
 {\includegraphics[width=0.41\textwidth]{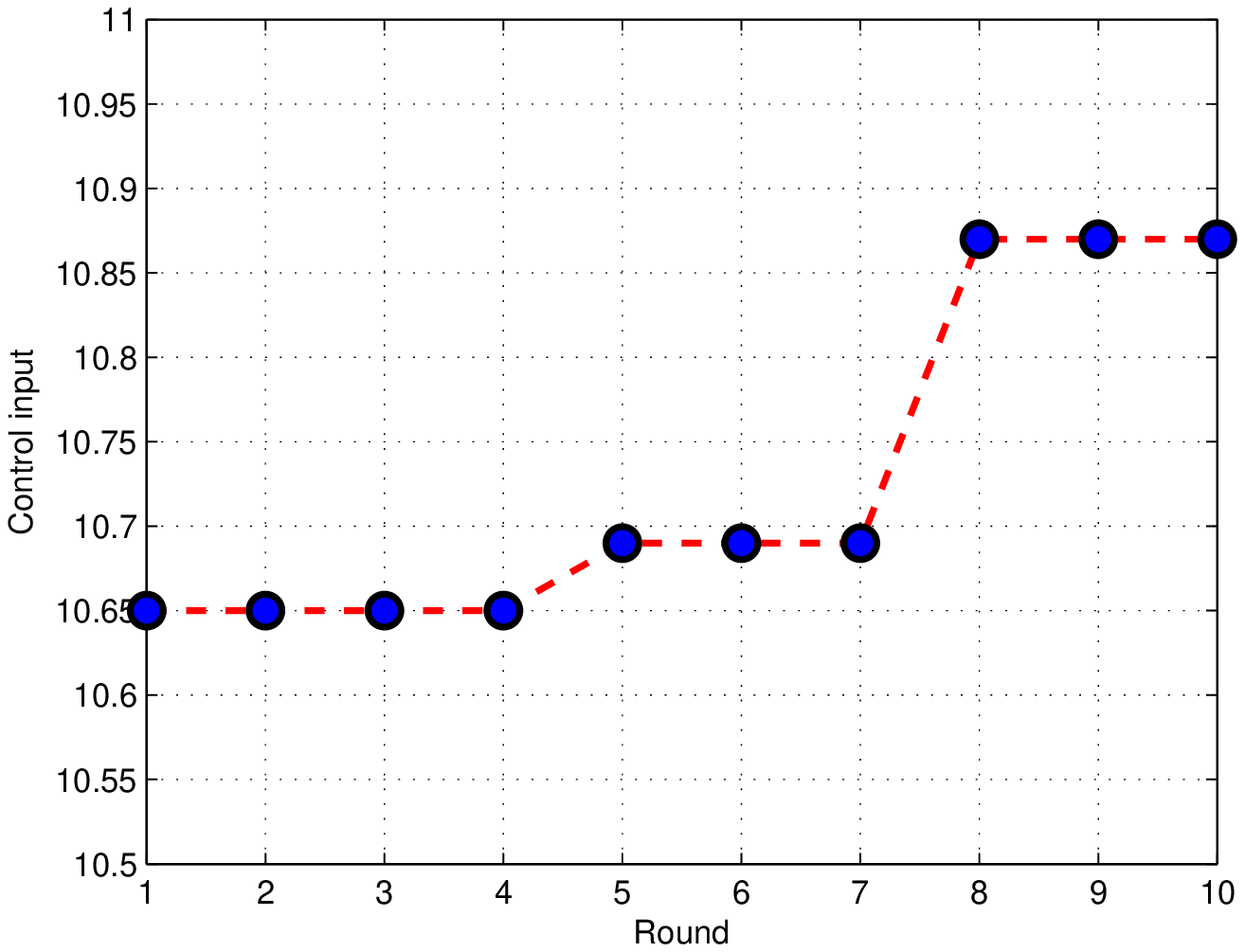}}
 {\includegraphics[width=0.41\textwidth]{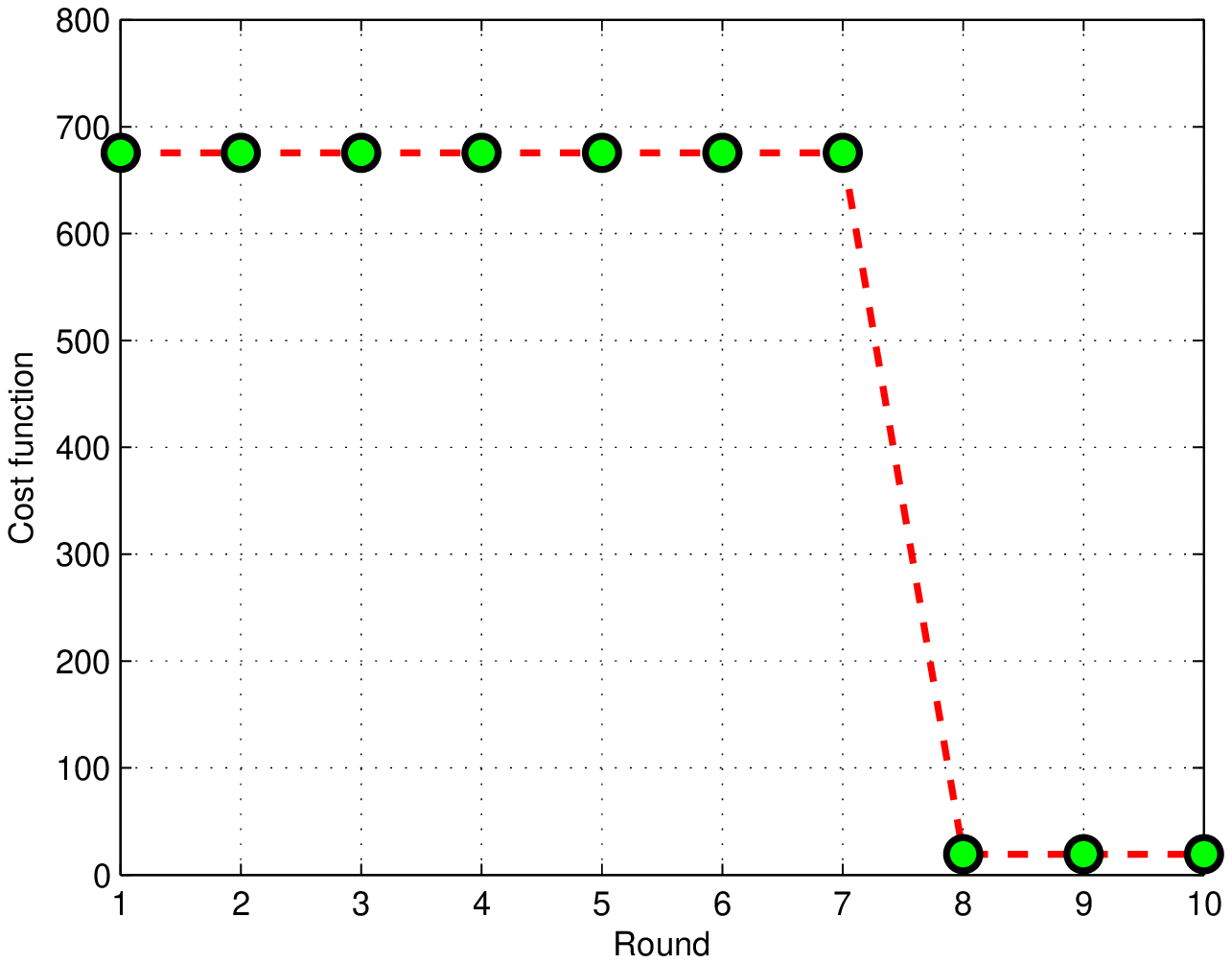}} \\
 \caption{\label{ju} Time evolution of control input and the resulted cost on Branch $2$.}
\end{figure}

Branch $2$ is selected to add the disruptive disturbance that initiates the chain reaction of cascading blackout. In Fig. \ref{ieee9bus}, red balls denote the generator buses, and green ones refer to the load buses. Bus identity (ID) numbers and branch ID numbers are marked as well. The arrows represent the power flow on each branch. A branch is severed once its transmission power exceeds the given threshold. The arrow disappears if there is no power transmission on the branch. The power system is running in the normal state at Step 1. Then the disruptive disturbance computed by the ISA (susceptance decrement 10.87) is added to sever Branch 2 at Step 2. Then Branch 1, Branch 4 and Branch 5 break off simultaneously at Step 3. Subsequently, Branch 3, Branch 6, Branch 7 and Branch 9 are removed from the power system at Step 4. As a result, the power network is divided into 8 islands without any power consumption. In particular, there is no power transmission on Branch 8 since Bus 6 and Bus 9 are both load buses.

\begin{figure}
\scalebox{0.65}[0.65]{\includegraphics{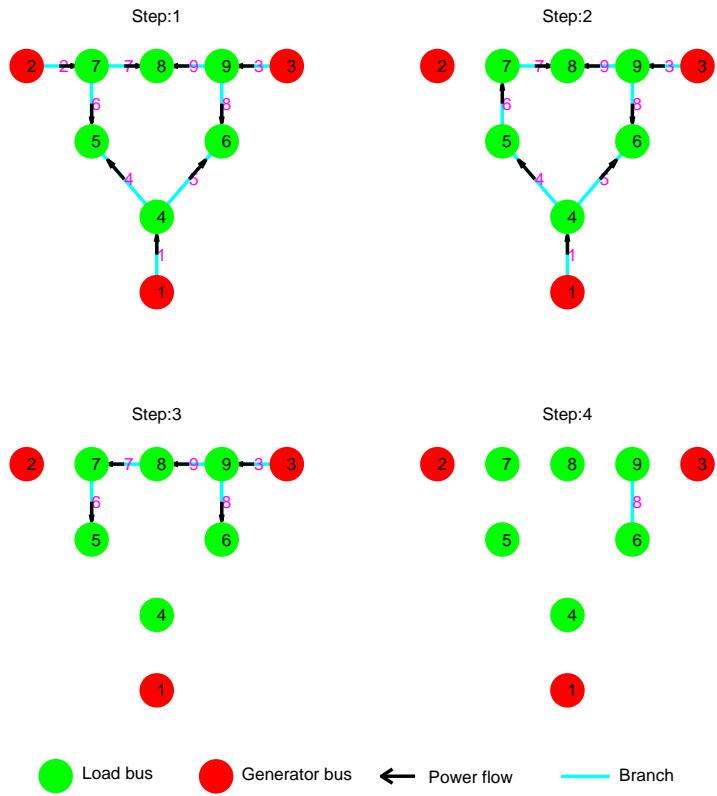}}\centering
\caption{\label{ieee9bus} Cascading process of the IEEE 9 Bus System under the computed initial disturbances on Branch 2.}
\end{figure}

\subsection{IEEE 14 Bus System}

\begin{table}
  \centering
  \caption{IEEE 14 Bus System-Branch data.}
  \label{tab:branch14}
  \begin{tabular}{ccccc}
    \toprule
    Branch number & Source bus & Sink bus & Reactance & Power threshold \\
    \midrule
      1 & 1 & 2 & 0.059 & 0.3 \\
      2 & 1 & 5 & 0.223 & 0.3 \\
      3 & 2 & 3 & 0.198 & 0.4 \\
      4 & 2 & 4 & 0.176 & 0.3 \\
      5 & 2 & 5 & 0.174 & 0.3 \\
      6 & 3 & 4 & 0.171 & 0.7 \\
      7 & 4 & 5 & 0.042 & 0.3 \\
      8 & 4 & 7 & 0.209 & 0.3 \\
      9 & 4 & 9 & 0.556 & 0.3 \\
      10 & 5 & 6 & 0.252 & 0.3 \\
      11 & 6 & 11 & 0.199 & 0.3 \\
      12 & 6 & 12 & 0.256 & 0.3 \\
      13 & 6 & 13 & 0.130 & 0.3 \\
      14 & 7 & 8 & 0.176 & 0.3 \\
      15 & 7 & 9 & 0.110 & 0.3 \\
      16 & 9 & 10 & 0.085 & 0.3 \\
      17 & 9 & 14 & 0.270 & 0.3 \\
      18 & 10 & 11 & 0.192 & 0.3 \\
      19 & 12 & 13 & 0.200 & 0.3 \\
      20 & 13 & 14 & 0.348 & 0.3 \\
    \bottomrule
  \end{tabular}
\end{table}

\begin{table}[h!]
  \centering
  \caption{IEEE 14 Bus System-Bus data.}
  \label{tab:bus14}
  \begin{tabular}{ccc}
    \toprule
    Bus number & Bus type & Power injection \\
    \midrule
      1 & R & 0  \\
      2 & G & 0.217  \\
      3 & G & 0.942  \\
      4 & L & -0.478  \\
      5 & L & -0.076  \\
      6 & G & 0.112  \\
      7 & L & 0  \\
      8 & G & 0  \\
      9 & L & -0.295  \\
     10 & L & -0.090  \\
     11 & L & -0.035  \\
     12 & L & -0.061  \\
     13 & L & -0.135  \\
     14 & L & -0.149  \\
    \bottomrule
  \end{tabular}
\end{table}

\begin{figure}
\scalebox{0.8}[0.8]{\includegraphics{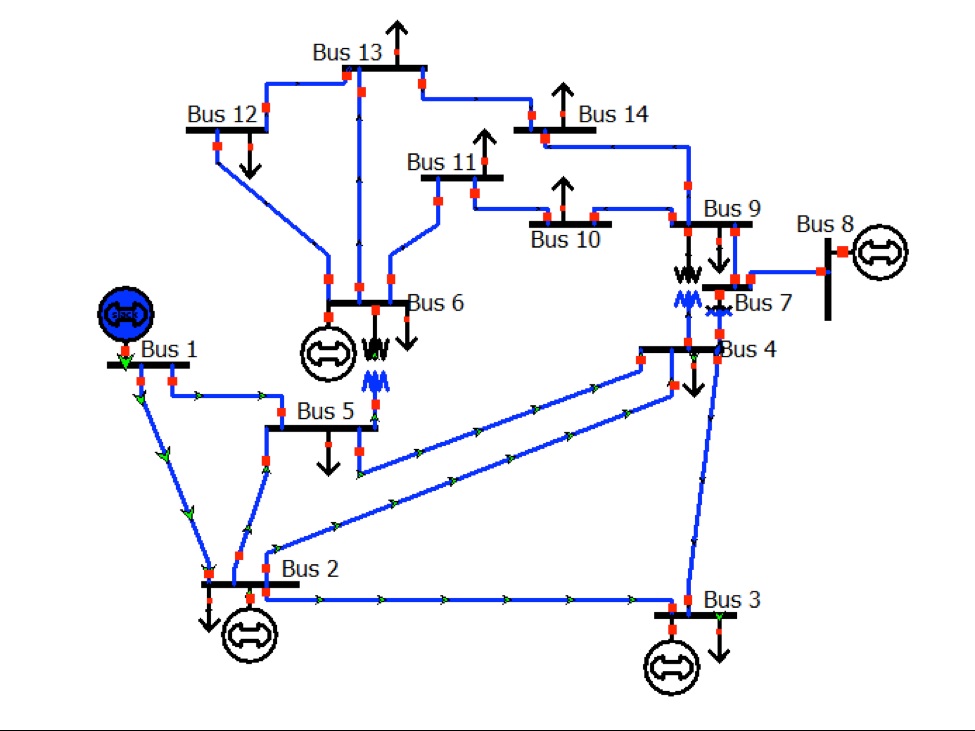}}\centering
\caption{\label{sim14bus} IEEE 14 Bus System.}
\end{figure}

\begin{figure}\centering
 {\includegraphics[width=0.48\textwidth]{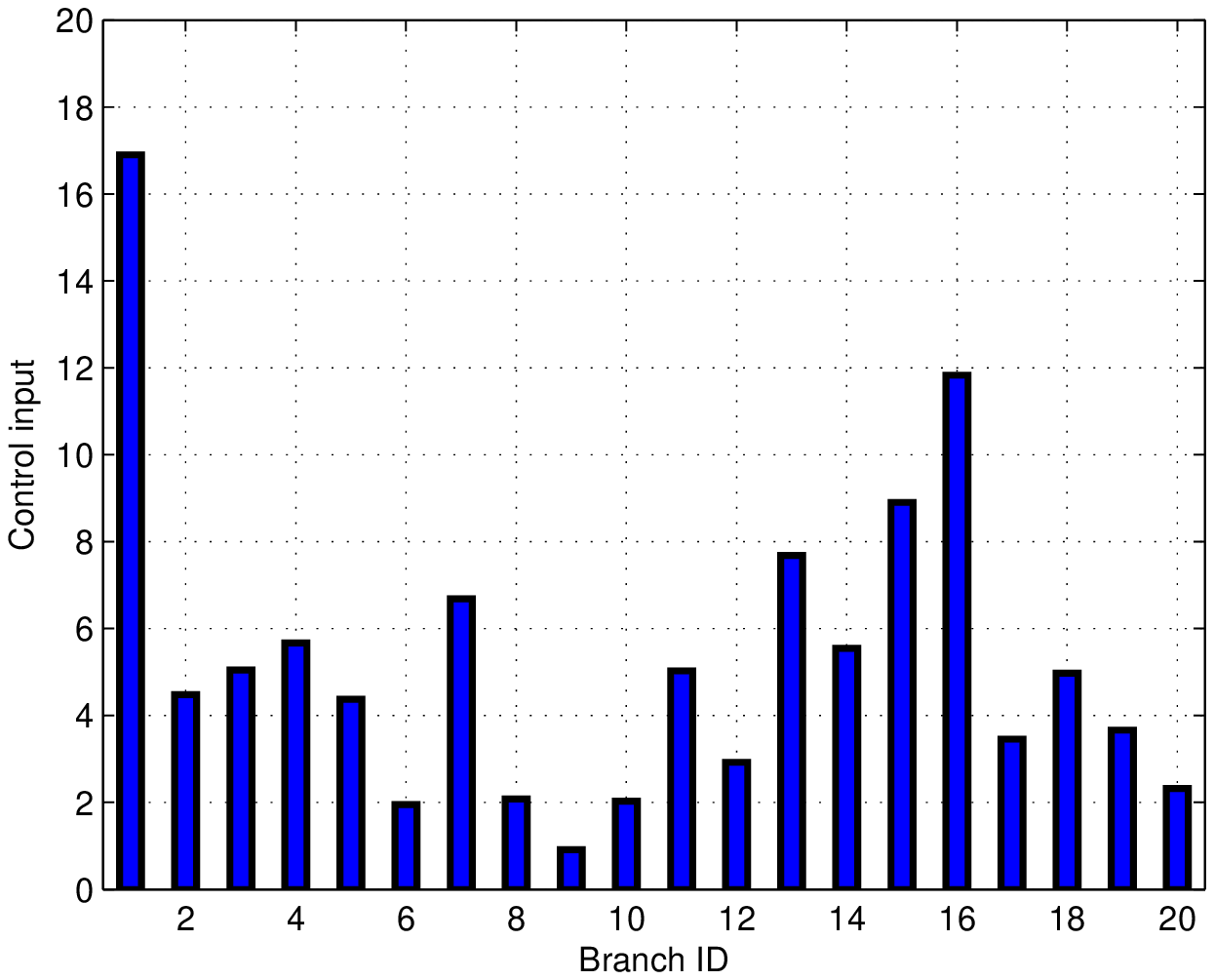}}
 {\includegraphics[width=0.48\textwidth]{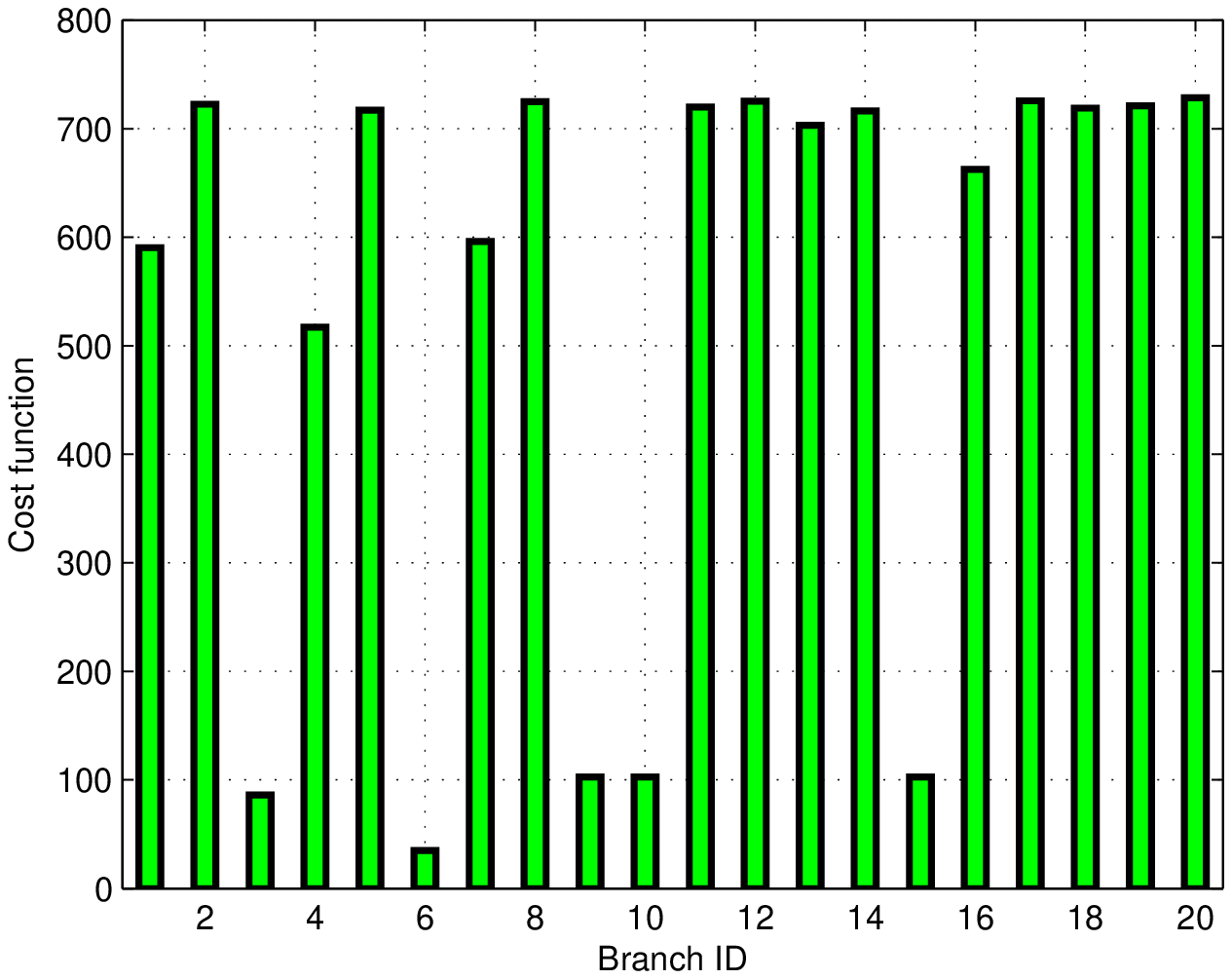}} \\
 \caption{\label{id14} Control input and the resulted cost on each branch of the IEEE 14 Bus System.}
\end{figure}

\begin{figure}\centering
 {\includegraphics[width=0.45\textwidth]{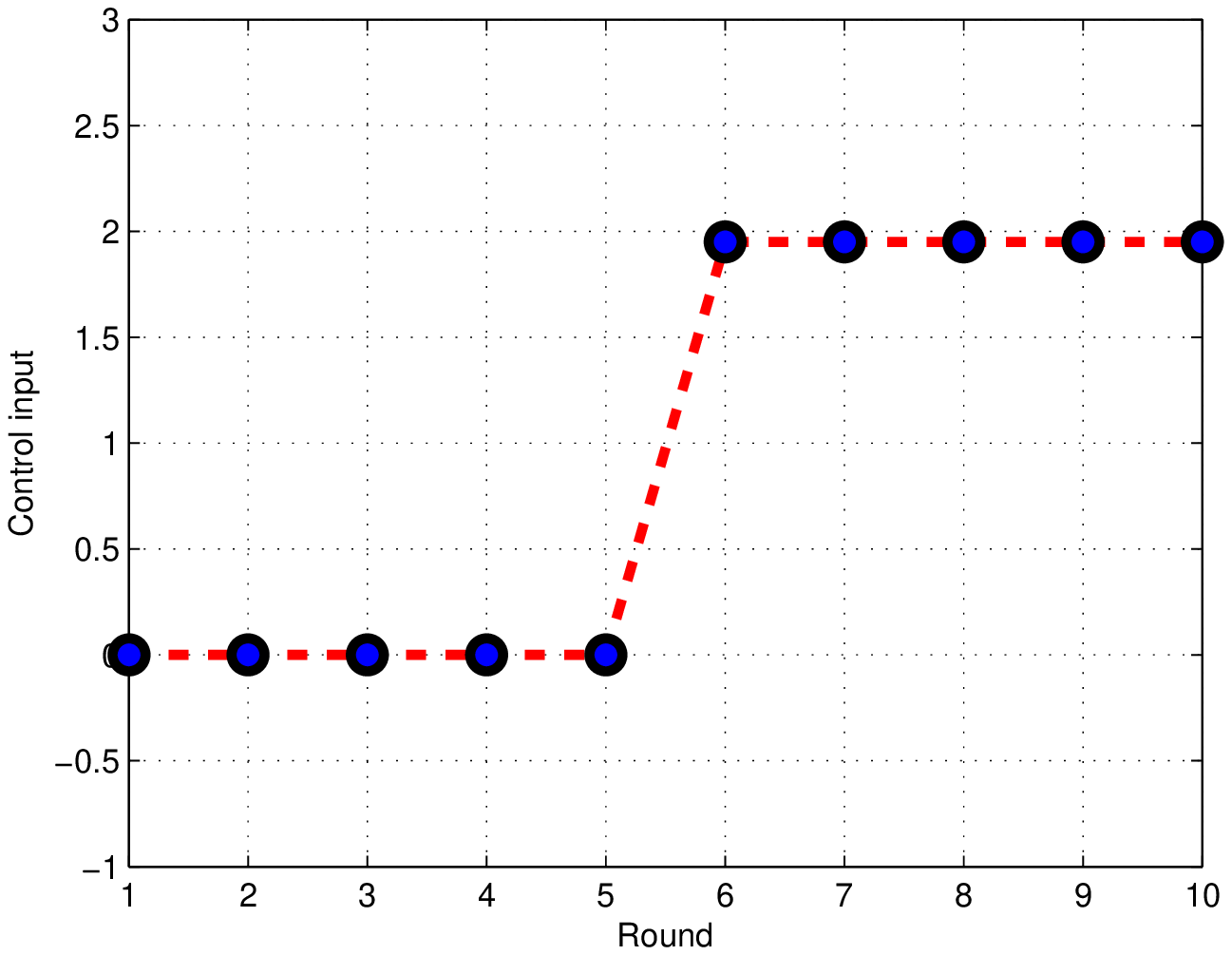}}
 {\includegraphics[width=0.45\textwidth]{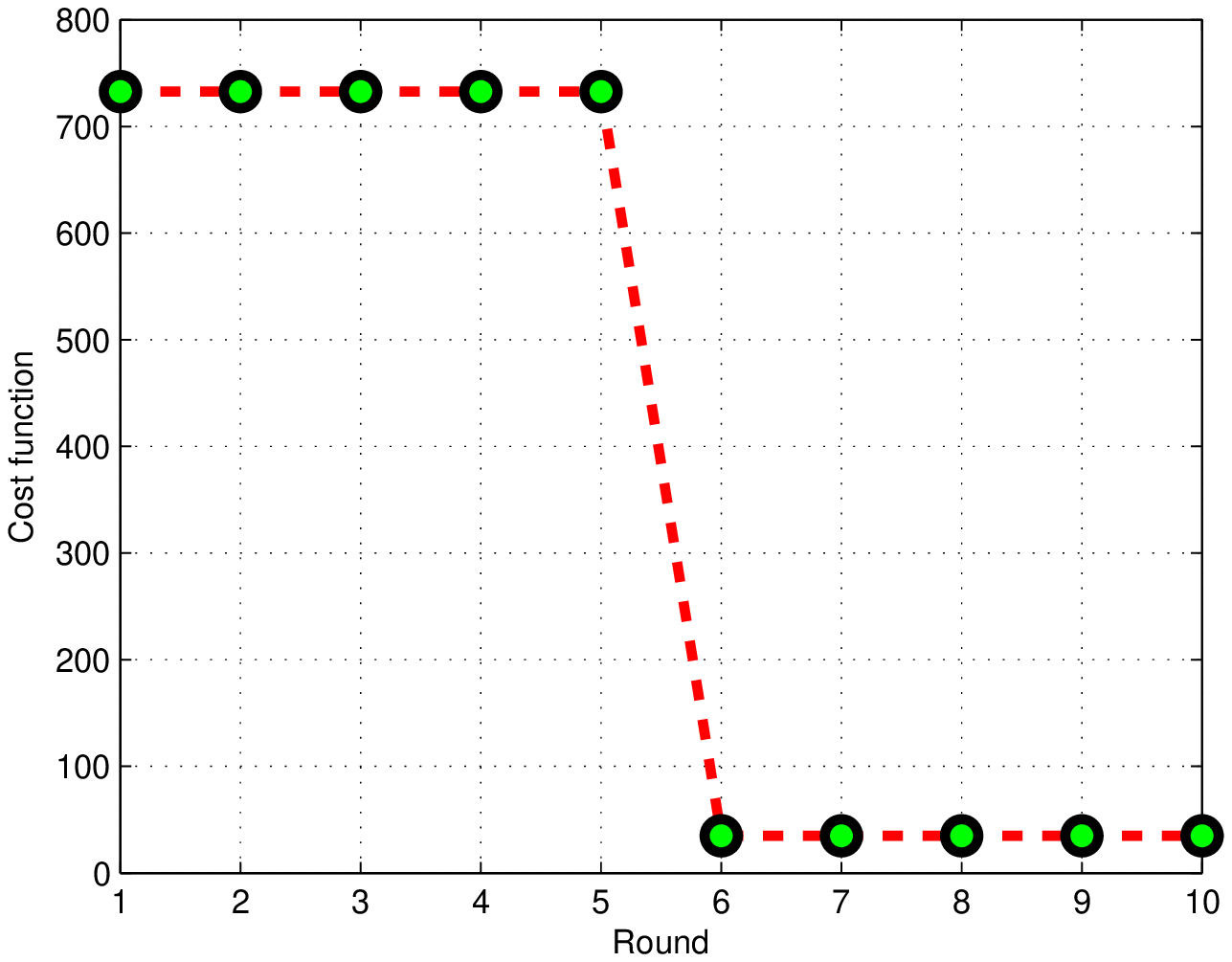}} \\
 \caption{\label{ju14} Time evolution of control input and the resulted cost on Branch $6$.}
\end{figure}

The ISA is also implemented on the IEEE 14 Bus System (see Fig. \ref{sim14bus}) to trace the initial disturbance on branches that result in the worst blackout of power network. The relevant branch data and bus data are shown in Table \ref{tab:branch14} and Table \ref{tab:bus14}, respectively \cite{zim11}. Other parameters for are given as follows: $\sigma=5\times 10^4$, $\epsilon=10^{-4}$, $\iota=1$, $i_{\max}=10$, $J_{\max}=10^6$ and $m=10$. Figure \ref{id14} presents the computed control input on each branch and the resulted cost level at the final step. Of all the computed disturbances, we can observe that the disturbance on Branch $6$ (red link) leads to the least value ($34.87$) of cost function, which implies the worst blackout of power networks. The process of iterative search for the least cost value and the corresponding control input is illustrated in Fig. \ref{ju14}. In particular, the cascading process caused by the initial admittance change of $1.95$ on Branch $6$ is shown in Fig. \ref{ieee14bus}. The process ends up with $2$ connected subnetworks and $8$ isolated buses after $6$ cascading steps. The subnetwork with one generator bus (Bus 6) and $3$ load buses (Bus 5, Bus 12 and Bus 13) is still in operation, while the other one with two load buses (Bus 9 and Bus 14) stops running due to the lack of power supply.

\begin{figure}
\scalebox{1.0}[1.0]{\includegraphics{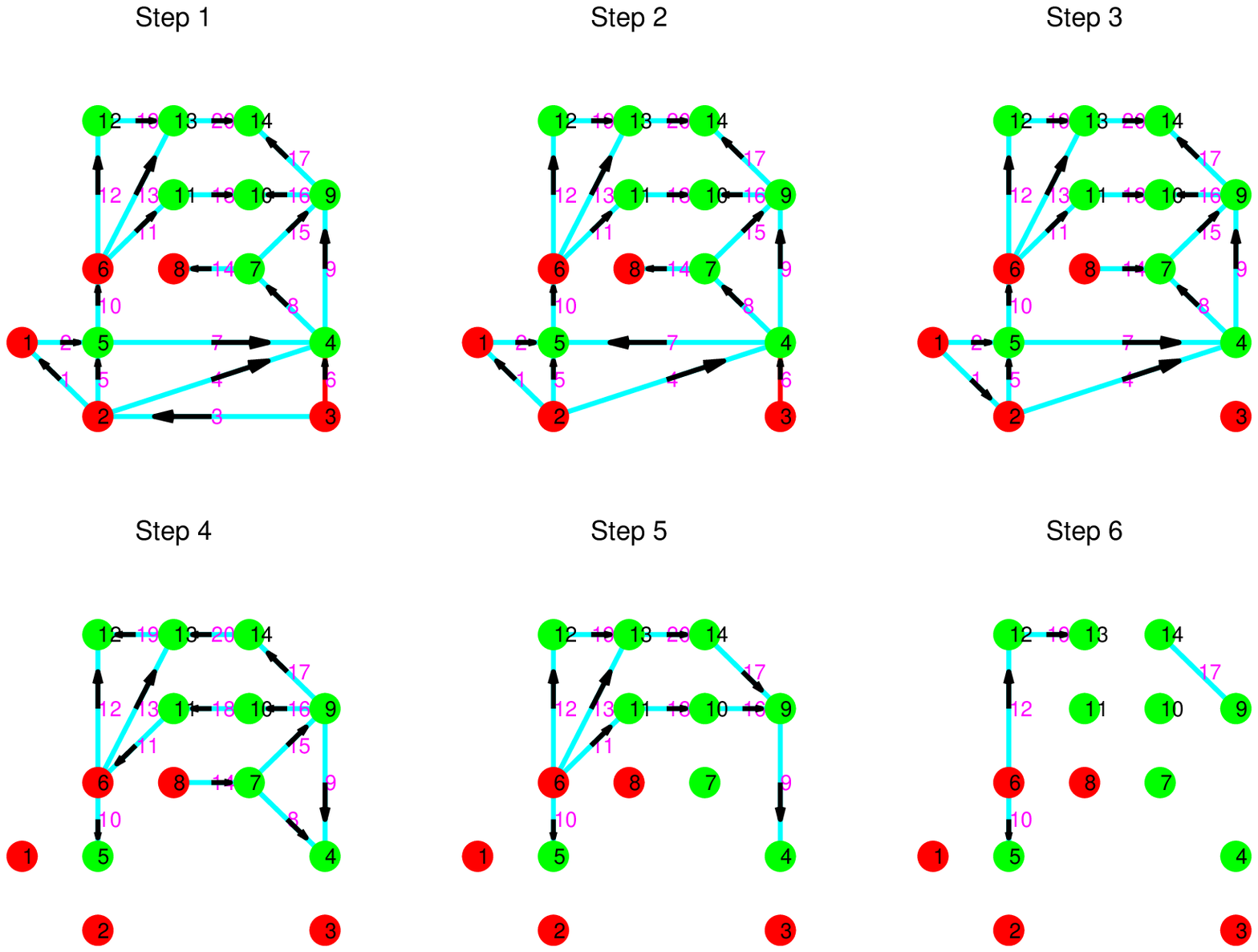}}\centering
\caption{\label{ieee14bus} Cascading process of the IEEE 14 Bus System under the computed initial  disturbances on Branch 6.}
\end{figure}

%\subsection{Discussion}

The validation results on IEEE 9 Bus System and IEEE 14 Bus System demonstrate the power network can be completely destroyed by disruptive disturbances on branches. In the simulations, the convergence rate of the ISA strongly depends on the initial condition of solving the system of algebraic equation (\ref{con_sys}) in each iteration.

\section{Conclusions}\label{sec:con}

A cascading model of transmission lines was developed to describe the evolution of branches on power systems under disruptive contingencies. With the cascading model and DC power flow equation, the identification problem of worst case cascading failures was formulated with the aid of optimal control theory by treating the disturbances as the control inputs. Simulation results demonstrate the effectiveness of our approach. The proposed approach allows us to determine the most disruptive disturbances on the targeted branch, which provides a new perspective of designing the corresponding protection strategy to enhance the resilience and stability of power system and interdependent critical infrastructure systems. Future work includes improving the cascading dynamics of power system with AC power flow equation and designing cooperative control strategies of protective relay to protect power systems \cite{hong11,zhai12}.

\section*{Acknowledgment}

This work is partially supported by the Future Resilience System Project at the Singapore-ETH Centre (SEC), which is funded by the National Research Foundation of Singapore (NRF) under its Campus for Research Excellence and Technological Enterprise (CREATE) program. It is also supported by Ministry of Education of Singapore under Contract MOE2016-T2-1-119.

\section*{Appendix}

%\subsection{Proof of Theorem \ref{sysAE}}

Now we present the proof of Theorem \ref{sysAE}.  From Theorem \ref{disopt}, the necessary conditions for the optimal control problem (\ref{cost}) can be determined as
\begin{equation}\label{cond_st}
Y_p^{k+1}=G(P_{ij}^k,c_{ij})\cdot Y_p^{k}+E_{i_k}u_k
\end{equation}

\begin{equation}\label{cond_u}
\left(\frac{\partial Y_p^{k+1}}{\partial u_k}\right)^T\lambda_{k+1}+\frac{\epsilon }{\max\{0,\iota-k\}}\cdot\frac{\partial\|u_k\|^2}{\partial u_k}=0
\end{equation}

\begin{equation}\label{cond_y}
\lambda_k=\left(\frac{\partial Y_p^{k+1}}{\partial Y_p^k}\right)^T\lambda_{k+1}+\frac{\epsilon }{\max\{0,\iota-k\}}\cdot\frac{\partial\|u_k\|^2}{\partial Y_p^k}
\end{equation}

\begin{equation}\label{cond_final}
\frac{\partial \mathrm{T}(Y_p^m)}{\partial Y_p^m}-\lambda_m=\mathbf{0}_n
\end{equation}
where $\mathbf{0}_n=(0,0,...,0)^T \in R^n$. Thus, solving Equation (\ref{cond_u}) leads to
\begin{equation}\label{cond_uk}
u_k=-E_{i_k}\frac{\lambda_{k+1}}{2\epsilon}\max\{0,\iota-k\}
\end{equation}
and simplifying Equation (\ref{cond_y}) yields
\begin{equation}\label{cond_lamb}
\lambda_k=\left(\frac{\partial Y_p^{k+1}}{\partial Y_p^k}\right)^T\lambda_{k+1}
\end{equation}
with the final condition $\lambda_m=\frac{\partial \mathrm{T}(Y_p^m)}{\partial Y_p^m}$ being derived from Equation (\ref{cond_final}).
Therefore, we have
\begin{equation}\label{cond_lambd}
\lambda_{k+1}=\prod_{s=0}^{m-k-2}\frac{\partial Y_p^{m-s}}{\partial Y_p^{m-s-1}}\cdot \frac{\partial \mathrm{T}(Y_p^m)}{\partial Y_p^m}.
\end{equation}
Combining Equations (\ref{cond_uk}) and (\ref{cond_lambd}), we obtain
\begin{equation}\label{cond_ukh}
u_k=-\frac{\max\{0,\iota-k\}}{2\epsilon}E_{i_k}\prod_{s=0}^{m-k-2}\frac{\partial Y_p^{m-s}}{\partial Y_p^{m-s-1}}\cdot \frac{\partial \mathrm{T}(Y_p^m)}{\partial Y_p^m}
\end{equation}
Substituting (\ref{cond_ukh}) into (\ref{cond_st}) yields
\begin{equation*}
Y_p^{k+1}-G(P^k_{ij},c_{ij})Y_p^k+\frac{\max\{0,\iota-k\}}{2\epsilon}E_{i_k}\prod_{s=0}^{m-k-2}\frac{\partial Y_p^{m-s}}{\partial Y_p^{m-s-1}}\cdot \frac{\partial \mathrm{T}(Y_p^m)}{\partial Y_p^m}=\mathbf{0}_n, \quad k=0,1,...,m-1
\end{equation*}
which is the integrated mathematical representation of necessary conditions (\ref{cond_st}), (\ref{cond_u}), (\ref{cond_y}) and (\ref{cond_final}) for the optimal control problem (\ref{cost}).

Next, we focus on the computation of the matrix
$$
\frac{\partial Y_p^{k+1}}{\partial Y_p^{k}}, \quad k=0,1,...,m-1
$$
Clearly, this matrix can be rewritten as
\begin{equation}\label{matrix}
\frac{\partial Y_p^{k+1}}{\partial Y_p^{k}}=\left(
                                              \begin{array}{cccc}
                                                \frac{\partial y_{p,1}^{k+1}}{\partial y_{p,1}^{k}}& \frac{\partial y_{p,1}^{k+1}}{\partial y_{p,2}^{k}} & . & \frac{\partial y_{p,1}^{k+1}}{\partial y_{p,n}^{k}} \\
                                                \frac{\partial y_{p,2}^{k+1}}{\partial y_{p,1}^{k}} & \frac{\partial y_{p,2}^{k+1}}{\partial y_{p,2}^{k}} & . & \frac{\partial y_{p,2}^{k+1}}{\partial y_{p,n}^{k}} \\
                                                . & . & . & . \\
                                                \frac{\partial y_{p,n}^{k+1}}{\partial y_{p,1}^{k}} & \frac{\partial y_{p,n}^{k+1}}{\partial y_{p,2}^{k}} & . & \frac{\partial y_{p,n}^{k+1}}{\partial y_{p,n}^{k}} \\
                                              \end{array}
                                            \right)
\end{equation}
where
$$
y_{p,l}^{k+1}=g(P^k_{i_lj_l},c_{i_lj_l})y_{p,l}^{k}+e_l^TE_{i_k}u_k.
$$
Therefore, we have
\begin{equation}\label{entry}
\begin{split}
\frac{\partial y_{p,l}^{k+1}}{\partial y_{p,s}^{k}}&=\frac{\partial g(P^k_{i_lj_l},c_{i_lj_l})}{\partial y_{p,s}^{k}}y_{p,l}^{k}+g(P^k_{i_lj_l},c_{i_lj_l})\frac{\partial y_{p,l}^{k}}{\partial y_{p,s}^{k}} \\
&=\frac{\partial g(P^k_{i_lj_l},c_{i_lj_l})}{\partial P^k_{i_lj_l}}\cdot\frac{\partial P^k_{i_lj_l}}{\partial y_{p,s}^{k}}y_{p,l}^{k}+g(P^k_{i_lj_l},c_{i_lj_l})\frac{\partial y_{p,l}^{k}}{\partial y_{p,s}^{k}}, \quad s,l=1,2,...,n
\end{split}
\end{equation}
where
\begin{equation}\label{term1}
\frac{\partial y_{p,l}^{k}}{\partial y_{p,s}^{k}}=\left\{
                                                  \begin{array}{ll}
                                                    1, & \hbox{$s=l$,} \\
                                                    0, & \hbox{$s\neq l$.}
                                                  \end{array}
                                                \right.
\end{equation}
and
\begin{equation}\label{term2}
\frac{\partial g(P^k_{i_lj_l},c_{i_lj_l})}{\partial P^k_{i_lj_l}}=\left\{
                                                                \begin{array}{ll}
                                                                  -P^k_{i_lj_l}\sigma\cos\sigma((P^k_{i_lj_l})^2-c_{i_lj_l}^2), & \hbox{$\sqrt{c_{i_lj_l}^2-\frac{\pi}{2\sigma}}<|P^k_{i_lj_l}|< \sqrt{c_{i_lj_l}^2+\frac{\pi}{2\sigma}}$;} \\
                                                                 0, &\hbox{otherwise.}
                                                                \end{array}
                                                              \right.
\end{equation}
It follows from Lemma \ref{lem_pij} and Lemma \ref{lem_der} that
\begin{equation}\label{term3}
\begin{split}
\frac{\partial P^k_{i_lj_l}}{\partial y_{p,s}^{k}}&=\frac{\partial \left[e_{i_l}^TA^T diag(Y^k_p)Ae_{j_l}(e_{i_l}-e_{j_l})^T(A^T diag(Y^k_p)A)^{-1^*}P^k\right]}{\partial y_{p,s}^{k}}\\
&=\frac{\partial \left[e_{i_l}^TA^T diag(Y^k_p)Ae_{j_l}\right]}{\partial y_{p,s}^{k}}(e_{i_l}-e_{j_l})^T(A^T diag(Y^k_p)A)^{-1^*}P^k\\
&+e_{i_l}^TA^T diag(Y^k_p)Ae_{j_l}\frac{\partial \left[(e_{i_l}-e_{j_l})^T(A^T diag(Y^k_p)A)^{-1^*}P^k\right]}{\partial y_{p,s}^{k}}\\
&=e_{i_l}^TA^T diag\left(\frac{\partial Y^k_p}{\partial y_{p,s}^{k}}\right)Ae_{j_l}(e_{i_l}-e_{j_l})^T(A^T diag(Y^k_p)A)^{-1^*}P^k\\
&+e_{i_l}^TA^T diag(Y^k_p)Ae_{j_l}(e_{i_l}-e_{j_l})^T\frac{\partial (A^T diag(Y^k_p)A)^{-1^*}}{\partial y_{p,s}^{k}}P^k\\
&=e_{i_l}^TA^T diag\left(e_s\right)Ae_{j_l}(e_{i_l}-e_{j_l})^T(A^T diag(Y^k_p)A)^{-1^*}P^k\\
&-e_{i_l}^TA^T diag(Y^k_p)Ae_{j_l}(e_{i_l}-e_{j_l})^T(A^T diag(Y^k_p)A)^{-1^*}(A^Tdiag(e_s)A)^*(A^T diag(Y^k_p)A)^{-1^*}P^k.\\
\end{split}
\end{equation}
Thus, each element in Matrix (\ref{matrix}) is explicitly expressed by Equation (\ref{entry}), which can be obtained by taking into account Equations (\ref{thre_fun}), (\ref{term1}), (\ref{term2}) and (\ref{term3}). This completes
the proof of Theorem \ref{sysAE}.


\begin{thebibliography}{100}

\bibitem{empg04}  Final Report on the August 14, 2003 Blackout in the United States and Canada. {\sl Electricity Markets and Policy Group Technical report}, US-Canada Power System Outage Task Force, 2004.

\bibitem{utce07} Final Report System Disturbance on 4 November 2006. {\sl Technical Report}, Union for the Co-ordination of Transmission of Electricity, 2007.

\bibitem{hine16} Hines, Paul DH, and Pooya Rezaei. Cascading Failures in Power Systems. Smart Grid Handbook, 2016.

\bibitem{dob07} Dobson, I., Carreras, B.A., Lynch, V.E., and Newman, D.E. Complex systems analysis of series of blackouts: cascading failure, critical points, and self-organization. {\sl Chaos: An Interdisciplinary Journal of Nonlinear Science} 17, 026103, 2007.

\bibitem{vai12} Vaiman, M., Bell, K., Chen, Y., et al. Risk assessment of cascading outages: methodologies and challenges. {\sl IEEE Transactions on Power Systems} 27 (2), 631–641, 2012.


\bibitem{hine10} Hines, P., Cotilla-Sanchez, E., and Blumsack, S. Do topological models provide good information about vulnerability in electric power networks ? {\sl Chaos: An Interdisciplinary Journal of Nonlinear Science} 20 (3), 033122, 2010.

\bibitem{yuy16} Yu, Y., Xiao, G., Zhou, J., Wang, Y., Wang, Z., Kurths, J., Schellnhuber, H. J., System crash as dynamics of complex networks. {\sl Proceedings of the National Academy of Sciences}, 201612094, 2016.

\bibitem{cate84} Cate, E.G., Hemmaplardh, K., Manke, J.W., and Gelopulos, D.P. Time frame notion and time response of the models in transient, mid-term and long-term stability programs. {\sl IEEE Transactions on Power Apparatus and Systems} PAS-103 (1), 143–151, 1984.

\bibitem{roy94} Roytelman, I. and Shahidehpour, S.M. A comprehensive long term dynamic simulation for power system recovery. {\sl IEEE Transactions on Power Systems} 9 (3), 1427–1433, 1994.

\bibitem{jia16} Song, J., Cotilla-Sanchez, E., Ghanavati, G., Hines, P. D. Dynamic modeling of cascading failure in power systems. {\sl IEEE Transactions on Power Systems}, 31(3): 2085-2095, 2016.

\bibitem{tae16} Kim, Taedong, Stephen J. Wright, Daniel Bienstock, and Sean Harnett, Analyzing vulnerability of power systems with continuous optimization formulations, {\sl IEEE Transactions on Network Science and Engineering} 3(3): 132-146, 2016.

\bibitem{tarsi70} Tarsi, David. Simultaneous Solution of line-out and open-end line-to-ground short circuits. {\sl IEEE Transactions on Power Apparatus and Systems}, 6(PAS-89): 1220-1225, 1970.

\bibitem{fol82} Perez LG, Flechsig AJ, Venkatasubramanian VA, Modeling the protective system for power system dynamic analysis, {\sl IEEE Transactions on Power Systems}, 9(4): 1963-1973, 1994.

\bibitem{stot09} Stott, B., Jardim, J. and Alsa\c{c}, O., DC power flow revisited. {\sl IEEE Transactions on Power Systems}, 24(3): 1290-1300, 2009.

\bibitem{stag68} Stagg, Glenn W., and Ahmed H. El-Abiad. Computer Methods in Power System Analysis. McGraw-Hill, 1968.

\bibitem{god13} Godsil, Chris, and Gordon F. Royle. Algebraic graph theory. Vol. 207. Springer Science \& Business Media, 2013.

\bibitem{fran95} Frank L. Lewis and Vassilis L. Syrmos. Optimal Control, 2nd Edition. Wiley-Interscience, 2nd edition, October 1995.

\bibitem{yeh06} Yeh, James. Real Analysis: Theory of Measure and Integration. World Scientific, 2006.

\bibitem{zim11} R. D. Zimmerman, C. E. Murillo-S\'anchez, and R. J. Thomas, Matpower: Steady-state operations, planning and analysis tools for power systems research and education, {\sl IEEE Transactions on Power Systems}, vol. 26, no. 1, pp. 12-19, Feb. 2011.

\bibitem{hong11} Hong, Y., Zhai, C., Dynamic coordination and distributed control design of multi-agent systems, {\sl Control Theory \& Applications}, 10, 028, 2011.

\bibitem{zhai12} Zhai, C., Hong, Y., Decentralized sweep coverage algorithm for uncertain region of multi-agent systems, {\sl Proceedings of American Control Conference}, Montr\'{e}al, Canada,  pp. 4522-4527, 2012.

\end{thebibliography}
\end{document}